\documentclass{article}
\usepackage[colorlinks=true,linkcolor=blue]{hyperref}
\usepackage{amsmath}
\usepackage{amsfonts}
\usepackage{amssymb}
\usepackage{amsthm}
\usepackage[utf8]{inputenc}
\usepackage{graphicx}
\usepackage{url}
\usepackage{array,multirow,makecell}
\usepackage[T1]{fontenc}
\usepackage{setspace}
\usepackage{accents}
\usepackage{color}
\usepackage[11pt]{extsizes}
\usepackage[top=2cm, bottom=2cm, left=2cm, right=2cm]{geometry}
\usepackage{systeme}
\usepackage{comment}
\usepackage{bm}
\usepackage{appendix}
\usepackage{pgfplots}
\usepackage{authblk}
\usepackage{enumitem,refcount}

\newtheorem{theorem}{Theorem}[section]

\newtheorem{lemma}[theorem]{Lemma}
\newtheorem{Proposition}[theorem]{Proposition}

\theoremstyle{definition}
\newtheorem{definition}[theorem]{Definition}
\newtheorem{remark}[theorem]{Remark}
\newtheorem{example}{Example}

\makeatletter
\newenvironment{countlist}[2][]
  {\begin{enumerate}[#1]
     \setcounter{countlist}{0}%
     \def\countname{#2}%
     \let\olditem\item
     \renewcommand{\item}{\stepcounter{countlist}\olditem}}
  {  \renewcommand{\@currentlabel}{\arabic{countlist}}%
     \label{\countname}%
   \end{enumerate}}
\makeatother

\newcommand{\q}[1]{``#1''}

\DeclareMathOperator{\diag}{diag}
\DeclareMathOperator{\sign}{sgn}
\DeclareMathOperator\supp{supp}

\newcounter{countlist}

\title{Convex ordering for stochastic control: the (path dependent) swing contracts case.}

\author[1]{Gilles Pagès}
\author[1,2]{Christian Yeo}
\affil[1]{\footnotesize Sorbonne Université, Laboratoire de Probabilités, Statistique et Modélisation, UMR 8001, case 158, 4, pl. Jussieu,
F-75252 Paris Cedex 5, France}
\affil[2]{\footnotesize Engie Global Markets, 1 place Samuel Champlain, 92400 Courbevoie, France}

\date{}

\newcommand{\introthmname}{}
\newtheorem{introthminn}{\introthmname}
\newenvironment{introthm}[1]
  {\renewcommand{\introthmname}{#1}\begin{introthminn}}
  {\end{introthminn}}

\swapnumbers

\pgfplotsset{width=8cm,compat=1.14}
\begin{document}
\maketitle

\begin{abstract}
\noindent
We investigate propagation of convexity and convex ordering on a typical discrete-time stochastic optimal control problem, namely the pricing of swing option. The dynamics of the underlying asset is modelled by the Euler scheme of a Brownian diffusion with affine drift, and convex volatility. We prove that the value function associated to the stochastic optimal control problem is a convex function of the underlying asset price. We also introduce a domination criterion offering insights into the functional monotonicity of the value function with respect to parameters of the underlying dynamics. We particularly focus on the one-dimensional setting where, by means of Stein's formula and regularization techniques, we show that the convexity assumption for the volatility dynamics can be relaxed with a semi-convexity assumption. Finally, to validate our results, we also conduct numerical illustrations.
\end{abstract}


\vspace{0.2cm}

\section*{Introduction}
This paper explores theoretical properties of the value function in a context of stochastic optimal control and in connection with convexity. To this end, we rely on \emph{convex ordering theory}. Denote by $\mathbb{L}^1_{\mathbb{R}^d}\big(\mathbb{P}\big)$ the space of $\mathbb{R}^d$-valued $\mathbb{P}$-integrable random vectors for some probability measure $\mathbb{P}$ and let $U, V \in \mathbb{L}^1_{\mathbb{R}^d}\big(\mathbb{P} \big)$ with respective distributions $\mu, \nu$. We say that \emph{$U$ is dominated for the convex order by $V$}, denoted $U \preceq_{cvx} V$, if, for every convex function $f : \mathbb{R}^d \to \mathbb{R}$, one has
\begin{equation}
\label{cvx_ord_def}
\mathbb{E}f(U) \le \mathbb{E}f(V)
\end{equation}
or, equivalently, that $\mu$ is dominated for the convex order by $\nu$ (denoted $\mu \preceq_{cvx} \nu$) if, for every convex function $f : \mathbb{R}^d \to \mathbb{R}$,
\begin{equation}
\label{cvx_ord}
\int_{\mathbb{R}^d}^{} f(\xi) \,d\mu(\xi)  \le \int_{\mathbb{R}^d}^{} f(\xi) \,d\nu(\xi).
\end{equation}
It can be shown that, in the preceding definition, we may restrict to Lipschitz continuous and convex functions (see Lemma \ref{rq_carac_ord_cvx}) or to convex functions with at most linear growth (see Lemma A.1 in \cite{10.1214/19-AIHP1014}). Besides, it is to note that the convex order definition implies that both distributions have the same expectation thanks to the convexity of functions $f(\xi)= \pm \xi_i$, $i=1:d$.

If $d=1$ and \eqref{cvx_ord_def} or \eqref{cvx_ord} only holds for every \emph{non-decreasing} (resp. \emph{non-increasing}) \emph{convex functions}, we speak of domination for the \emph{non-decreasing} (resp. \emph{non-increasing}) \emph{convex order} denoted by $U \preceq_{ivx} V$ (resp. $U \preceq_{dvx} V$). Besides, the preceding definition of convex ordering is consistent in the sense that, for any integrable $\mathbb{R}^d$-valued random vector $U$ and for any convex function $f : \mathbb{R}^d \to \mathbb{R}$, one has $\mathbb{E}f(U) \in (-\infty, +\infty]$ (see Appendix \eqref{rq_consistency} for a proof).

Convex ordering-based approaches have recently been used to compare (possibly path-dependent) European option prices \cite{Bergenthu} or American option prices in local volatility models \cite{Pages_cvx_ord_path_dep}. Similar results have been obtained for the expectation of Mckean Vlasov processes \cite{liu2022functional, LIU2022312} (with an application to mean-field games) and more recently to compare solutions of Volterra equations \cite{jourdain:hal-03862241} with possibly singular kernels. Our paper builds upon and follows the numerous fairly recent studies on the usage of convex ordering theory. In our case, we intend to extend the latter type of results to Stochastic Optimal Control problems. To this end, we focus on \emph{swing contracts}. This commodity derivative product enables its holder to purchase amounts of energy $q_k$, at predetermined exercise dates, $t_k = \frac{kT}{n}$ ($k = 0, \ldots, n-1$), until the contract maturity at time $t_n = T$. The purchase price, or \emph{strike price}, at each exercise date, is denoted by $K_k$ and can be either a constant value (i.e., $K_k = K$, where $k = 0, \ldots, n-1$) or indexed on either the same commodity or another commodity past/future prices. The holder of the swing contract ought to purchase, at every time $t_k$, a quantity $q_k$ of the commodity subject to \emph{local constraints} i.e.
$$\underline{q} \le q_{k} \le \overline{q}, \quad k \in \{0,\ldots,n-1\}.$$
There also exists a \emph{global constraint} meaning that, at the maturity of the contract:
    \begin{equation}
    Q_{n} = \sum_{k = 0}^{n-1} q_{k} \in \big[\underline{Q}, \overline{Q}\big], \quad \text{with} \quad Q_0 = 0 \quad \text{and} \quad 0 \le \underline{Q} \le \overline{Q} < +\infty.
    \end{equation}
Such a swing contract is called a \emph{take-or-pay} contract, where all constraints are firm. Besides, there exists an alternative setting where the holder of the contract has to respect local constraints but in case of violation of global constraints, she has to pay a penalty, at the expiry date $t_n = T$, that is proportional to the excess/deficit of consumption.

For the two volume constraint settings aforementioned, at each exercise date $t_k$, the reachable cumulative consumption $Q_k = \sum_{\ell = 0}^{k-1} q_\ell$ (with $Q_0  =0$) lies within the set $\mathcal{Q}_m(t_k)$, for $m \in \{\emph{firm}, \emph{pen}\}$, defined by (we assume $\underline{q} = 0$ since one may always be reduced to this case \cite{Bardou2009OptimalQF})
\begin{equation}
\mathcal{Q}_{\emph{firm}}(t_k) := \Big[\underbrace{\max\big(0, \underline{Q} - (n-k) \cdot \overline{q} \big)}_{Q^d(t_{k})}, \underbrace{\min\big(k \cdot \overline{q}, \overline{Q}\big)}_{Q^u(t_{k})}\Big] \quad \text{and} \quad\mathcal{Q}_{\emph{pen}}(t_k) := \big[0, k \cdot \overline{q}\big].
\end{equation}

The pricing of swing contracts appears as a (discrete time) stochastic optimal control problem with constraints, where the sequence $(q_k)_{0 \le k \le n-1}$ represents the control. As such, this contract can be evaluated through the \emph{Backward Dynamic Programming Equation}. To fix the probabilistic framework, we assume that the underlying asset (generally a forward contract) has a price $F_{t_k}$ at time $t_k$ that can be expressed as:
\begin{equation}
\label{def_fwd_gen}
F_{t_k} = f\big(t_k, X_{t_k}\big), \quad k \in \{0,\ldots,n\},
\end{equation}
where $f:[0, T] \times \mathbb{R}^d \to \mathbb{R}_{+}$ is a Borel function and $\big(X_{t_k}\big)_{0 \le k \le n}$ is an $\mathbb{R}^d$-valued Markov process. We also consider a filtered probability space $\big(\Omega, \{\mathcal{F}_{t_k}^X, 0 \le k \le n\}, \mathbb{P}\big)$, where $\big(\mathcal{F}_{t_k}^X, 0 \le k \le n\big)$ is the natural (completed) filtration of $(X_{t_k})_{0 \le k \le n}$ satisfying usual conditions. Then, we assume that the decision process $(q_{k})_{0 \le k \le n-1}$ is defined on the same probability space and is $\mathcal{F}_{t_k}^X$-adapted.

At each time $t_k$, by purchasing an amount $q \ge 0$, the holder of the swing contract makes an algebraic profit given by the following payoff function:
\begin{align}
\label{gen_payoff}
\Psi_k \colon [0, T] \times \mathbb{R}_{+} \times (\mathbb{R}^d)^{k+1}  &\to \mathbb{R} \nonumber \\
	(t_k, q, x_{0:k}) &\mapsto \Psi_k\big(t_k, q, x_{0:k}\big),
\end{align}
where we used the notation $x_{0:k} = (x_0, \ldots, x_k) \in (\mathbb{R}^d)^{k+1}$ which, at the same time, emphasizes the fact that we enlarge the class of payoff functions to include path-dependency as it could be for actual swing contracts. Then, under mild assumptions, for $k \in \{0,\ldots,n\}$, $c \in \{\emph{firm}, \emph{pen}\}$ and $Q_k \in \mathcal{Q}_c(t_k)$, one may prove (see \cite{Bardou2007WhenAS, BarreraEsteve2006NumericalMF}) that the swing price is given by the resolution of the following backward equation:
\begin{equation}
    \left\{
    \begin{array}{ll}
        v_k\big(x_{0:k}, Q_k\big) = \underset{q \in \mathbb{A}_c(t_k, Q_k)}{\sup} \hspace{0.1cm}\Bigg[ \Psi_k\big(t_{k}, q, x_{0:k}\big) + \mathbb{E}\Big(v_{k+1}\big(x_{0:k}, X_{t_{k+1}}, Q_k+q \big) \big\rvert X_{t_{0:k}} = x_{0:k} \Big)  \Bigg],\\
        \vspace{0.2cm}
        v_{n}\big(x_{0:n}, Q_{n}\big) = P_c\big(t_n, x_{0:n}, Q_n\big),
    \end{array}
\right.
\label{eq_dp_swing_generic}
\end{equation}
where $X_{t_{0:k}} := \big(X_{t_0},\ldots,X_{t_k}\big)$ and $\mathbb{A}_c(t_k, Q_k)$ for $c \in \{\emph{firm}, \emph{pen}\}$ is the set of admissible controls at time $t_k$ depending on the cumulative consumption $Q_k$ up to time $t_{k-1}$. As already mentioned, the index $c$ designates the constraint type, so that the preceding set is defined by:
\begin{equation}
    \left\{
    \begin{array}{ll}
       \mathbb{A}_{\emph{firm}}(t_k, Q_k) := \big[\max\big(0, Q^{d}({t_{k+1}}) - Q_k\big), \min\big(\overline{q}, Q^{u}({t_{k+1}}) - Q_k\big) \big],\\
        \\
        \mathbb{A}_{\emph{pen}}(t_k, Q_k) := [0, \overline{q}].
    \end{array}
\right.
\label{set_adm}
\end{equation}
In Equation \eqref{eq_dp_swing_generic}, the function $P_c$, defined on $[0, T] \times (\mathbb{R}^d)^{n+1} \times \mathbb{R}_{+}$, corresponds to the penalty function:
\begin{equation}
    \left\{
    \begin{array}{ll}
       P_{\emph{firm}}\big(t_n, x_{0:n}, Q_n\big) = 0,\\
        P_{\emph{pen}}\big(t_n, x_{0:n}, Q_n\big) := -f\big(t_n, x_{0:n}\big) \cdot \Big(A \cdot \big(Q_n - \underline{Q}\big)_{-} + B \cdot \big(Q_n - \overline{Q}\big)_{+}  \Big),
    \end{array}
\right.
\label{penalty_func}
\end{equation}
where $A$ and $B$ are positive real constants.

Our stochastic optimal control problem is discrete by nature. In this context, we suppose that $\big(X_{t_k}\big)_{0 \le k \le n}$ are the values at times $t_k$ of the $\mathbb{R}^d$-valued Brownian diffusion process with affine drift, solution of the following stochastic differential equation:
\begin{equation}
    \label{general_diffusion_X}
    X_t^x = x + \int_{0}^{t} \kappa(s) \big(X_s^x - \zeta\big) \,ds + \int_{0}^{t} \sigma_s(X_s^x) \,dW_s, \quad x \in \mathbb{R}^d,
\end{equation}
where $(W_t)_{t \in [0, T]}$ is a $q$-dimensional Brownian motion, $\kappa : [0,T] \to \mathbb{R}$ is a continuous function, and $\zeta \in \mathbb{R}^d$. Besides, we also adopt the notation $\sigma_t$ for the volatility field to mean:
$$\sigma_t := [0, T] \times \mathbb{R}^d \ni (t, x) \mapsto \sigma(t, x) \in \mathbb{M}_{d,q}(\mathbb{R})$$
and we suppose that it satisfies the standard space Lipschitz assumption, uniform in time:
\begin{equation}
    \label{hyp_strong_sol}
    \exists K> 0, \quad \forall t \in [0, T], \forall x, y \in \mathbb{R}^d, \quad  \|\sigma_t(x) - \sigma_t(y)\|\le K |x-y|
\end{equation}
and
\begin{equation}
    \underset{t \in [0, T]}{\sup} \hspace{0.1cm} \|\sigma_t(0)\| < + \infty,
\end{equation}
where for convenience, we assume that $\|\cdot \|$ is the Frobenius norm on $\mathbb{M}_{d, q}(\mathbb{R})$ (see notations further on).

With $X_{t_k}$ given by \eqref{general_diffusion_X} at exercise date $t_k$, we still consider the \emph{BDPP} \eqref{eq_dp_swing_generic}. Then, denote by $v_k^{[\sigma]}$ the swing value function in the \emph{BDPP} at exercise date $t_k$, where the superscript $[\sigma]$ is to emphasize the volatility of the underlying process $X_{t_k}$. To state typical results established in this paper, we need to introduce the following pre-order on $\mathbb{M}_{d, q}(\mathbb{R})$, namely
$$A \preceq B \hspace{0.2cm} \text{if} \hspace{0.2cm} BB^\top-AA^\top \hspace{0.3cm} \text{is a non-negative matrix.}$$ 
We then prove the following results.

\begin{introthm}{Theorem}
\label{main_thm_intro}
    Suppose that for any $c \in \{firm, pen\}$, $k \in \{0,\ldots,n-1\}$, $Q_n \in \mathcal{Q}_c(t_n)$ and $q \in [0, \overline{q}]$, the functions $(\mathbb{R}^d)^{k+1} \ni x_{0:k} \mapsto \Psi_k(t_k, q, x_{0:k})$, $(\mathbb{R}^d)^{n+1} \ni x_{0:n} \mapsto P_c(t_n, x_{0:n}, Q_n)$ are convex and Lipschitz and continuous. Then, one has the following two results.
    
\vspace{0.1cm}
\noindent
\textbf{(P1). [Convexity propagation]} If for all $t \in [0,T]$, $\sigma_{t}$ is \q{convex} (where the convexity should be understood in the sense of Definition \ref{def_cvx_matrix} further on, and linked to the pre-order defined above), then, at any time $t_k$, the swing price function $(\mathbb{R}^d)^{k+1} \ni x_{0:k} \mapsto v_k(x_{0:k}, Q_k)$ is convex.

\vspace{0.2cm}
\noindent
\textbf{(P2). [Domination criterion]} Consider two volatility functions $\sigma_{t}, \theta_{t}: \mathbb{R}^d \to \mathbb{R}_{+}$ for $t \in [0, T]$. If either $\sigma_{t}$ or $\theta_{t}$ is \q{convex} for all $t \in [0, T]$, then one has the following comparison result:
$$\big(\forall t \in [0, T], \hspace{0.1cm} \forall x \in \mathbb{R}^d, \quad \sigma_{t}(x)  \preceq \theta_{t}(x) \big)  \implies \Big(\forall x \in \mathbb{R}^d, k \in \{0,\ldots,n\}, Q_k \in \mathcal{Q}_c(t_k), \quad v_k^{[\sigma]}(x, Q_k) \le v_k^{[\theta]}(x, Q_k) \Big).$$
\end{introthm}

Our aim is to prove this theorem in a constructive way using an approach based on simulable approximations of the underlying state diffusion process $(X_t)_{t \in [0, T]}$ so that the stochastic control problem associated to these approximations also share the same properties as in Theorem \ref{main_thm_intro}. For options like European options for instance, the propagation of convexity mentioned in Theorem \ref{main_thm_intro} could follow owing to the partial differential equation (\emph{PDE}) based approach in \cite{LIONS2006915}. But here, in our discrete-time problem, we propose a probabilistic proof of the propagation of convexity for swing prices as well as for their numerical approximation schemes. Besides, to the best of our knowledge it does not exist a probabilistic proof of the \emph{PDE} approach in \cite{LIONS2006915}. For practitioners, the numerical scheme aspect, we described above, is crucial to derive efficient and robust numerical methods for the computation of the swing option prices. In particular, it is not clear at all that the characterization of convexity propagation provided in \cite{LIONS2006915} applies to these approximation schemes and anyway they only cover the \q{vanilla} case of marginal functionals of the form $\mathbb{E}f\big(X_t\big), t \in [0, T]$.

Both results in the preceding theorem are derived using the \emph{Backward Dynamic Programming Principle} (\emph{BDPP}) presented in \eqref{eq_dp_swing_generic}. Convex ordering theory will intervene in the propagation of convexity through what is often called the \emph{continuation value} which is represented by the conditional expectation involved in the \emph{BDPP}. Besides, an important part of this paper will be dedicated to the one-dimensional setting ($d=q=1$), where we show that the convexity of the volatility function $\sigma_{t}$ may be relaxed with a semi-convexity assumption.

\vspace{0.4cm}
\noindent
\textbf{\textsl{Notations}.}
$\bullet$ $\mathbb{R}^d$ is equipped with the canonical Euclidean norm denoted by $|\cdot|$. 

\noindent
$\bullet$ For all $x=(x_1, \ldots, x_d), y=(y_1,\ldots,y_d) \in \mathbb{R}^d$, $\langle x, y \rangle = \sum_{i = 1}^d x_iy_i$ denotes the canonical inner (Euclidean) product. 

\noindent
$\bullet$ $x_{0:k}$ denotes the vector $(x_0, \ldots, x_k)$. Then for any $k\in \mathbb{N}$, $(\mathbb{R}^d)^{k+1}$ is equipped with the norm:
$$\forall x_{0:k} \in (\mathbb{R}^d)^{k+1}, \quad \big|x_{0:k}\big|_{(k)} := \max\big(|x_0|, \ldots, |x_k|\big).$$
$\bullet$ $\mathbb{M}_{d,q}\big(\mathbb{R}\big)$ denotes the space of real-valued matrix with $d$ rows and $q$ columns and is equipped with either the classical Fröbenius norm or the operator norm defined respectively by,
$$\|A\|_F = \sqrt{Tr(AA^\top)} = \sqrt{\sum_{i = 1}^{d} \sum_{j = 1}^{q} a_{i, j}^2} \quad \text{and} \quad \|A\|_{op} = \underset{|x| = 1}{\sup} \hspace{0.1cm} \big|Ax\big|.$$
$\bullet$ $\mathcal{B}(0, R)$ denotes the closed ball with radius $R \ge 1$ defined by $\{x \in \mathbb{R}^d : |x| \le R\}$. We also define the uniform norm on this ball: for any continuous function $f : \mathbb{R}^d \to \mathbb{R}$:
$$\big\|f \big\|_{\mathcal{B}(0, R)} := \underset{x \in \mathcal{B}(0, R)}{\sup} \hspace{0.1cm} |f(x)|.$$

\noindent
$\bullet$ $\mathcal{S}^{+}\big(q, \mathbb{R}\big)$ and $\mathcal{O}\big(q, \mathbb{R}\big)$ denote respectively the subsets of $\mathbb{M}_{q, q}\big(\mathbb{R}\big)$ of symmetric positive semi-definite and orthogonal matrices with real entries. 

\noindent
$\bullet$ For a countable set $E$, $|E|$ denotes its cardinality. 

\noindent
$\bullet$ $\mathcal{C}^2\big(\mathbb{R}\big)$ will denotes the set of real-valued twice continuously differentiable functions.

\noindent
$\bullet$ $f \ast g$ is the convolution product on $\mathbb{R}$ between two functions $f, g : \mathbb{K} \to \mathbb{K}$, defined by
$$\big(f \ast g \big)(x) := \int_{\mathbb{K}}^{} f(t) \cdot g(x-t) \, \mathrm{d}t = \int_{\mathbb{K}}^{} f(x-t) \cdot g(t) \, \mathrm{d}t.$$

\noindent
$\bullet$ $\mathcal{P}_p\big(\mathbb{R}^d\big)$ denotes the set of probability distributions on $\mathbb{R}^d$ with $p^{th}$ finite moment. 

\noindent
$\bullet$ For all $x \in \mathbb{R}$, $\sign(x) := \mathbf{1}_{x > 0} - \mathbf{1}_{x < 0}$ and $\overline{\sign}(x) := \mathbf{1}_{x \ge 0} - \mathbf{1}_{x < 0}$.

\numberwithin{equation}{section}

\section{Ordering of swing contract values}
\label{swing_cont}

\subsection{Short background on convex ordering}
We start this section by some preliminaries following the definition of convex ordering. This will serve as basis to establish our main results. We refer the reader to Appendix \ref{appendix_proof_some_prop} for a proof.

\begin{Proposition}
\label{prop_first_prop_cvx_ord}
\begin{countlist}[label={(\alph*)}]{otherlist}
    \item \label{coherent_cvx_ord_varriance} Let $U, V \in \mathbb{L}^2_{\mathbb{R}^d}\big( \mathbb{P} \big)$ and define $\mathbb{V}ar(U) := \mathbb{E}|U|^2 - |\mathbb{E}U|^2$. Then, by setting $f : x \mapsto x_i^2$, $i=1:d$, one has,
	$$U \preceq_{cvx} V \implies \mathbb{V}ar(U) \le \mathbb{V}ar(V).$$  
  
  \item \label{comp_gauss_cvx_ord} (\textbf{Gaussian distributions (centered)}): Let $Z \sim \mathcal{N}(0,\mathbf{I}_q)$ on $\mathbb{R}^q$ and let $A, B \in \mathbb{M}_{d, q}\big(\mathbb{R}\big)$, then
$$BB^\top - AA^\top \in \mathcal{S}^{+}\big(d, \mathbb{R}\big) \implies AZ \preceq_{cvx} BZ$$

	\noindent
	or equivalently $\mathcal{N}(0,AA^\top) \preceq_{cvx} \mathcal{N}(0,BB^\top)$. In particular, if $d = q = 1$, one has
	$$|\sigma| \le |\vartheta| \implies \mathcal{N}(0,\sigma^2) \preceq_{cvx} \mathcal{N}(0,\vartheta^2).$$
\end{countlist}
\end{Proposition}

\begin{remark}[Convex ordering and risk measure]
Convex ordering has been widely used in actuarial science and quantitative finance to quantify or compare risk through the notion of risk measure \cite{BURGERT2006289, BAUERLE2006132, Gupta_cvx_actuarial_science}. In particular, one may show that convex ordering is consistent with certain risk measures on certain probability spaces \cite{BAUERLE2006132}. The latter means that for two integrable random variables $X, Y$ (which may represent potential losses) and an appropriate risk measure $\rho : \mathbb{L}_{\mathbb{R}}^1\big(\Omega, \mathcal{A}, \mathbb{P}\big) \to \mathbb{R}$ for some probability space $\big(\Omega, \mathcal{A}, \mathbb{P}\big)$, one has,
$$X \preceq_{cvx} Y \quad\text{or} \quad X  \preceq_{icv} Y \implies \rho(X) \le \rho(Y).$$
\noindent
Following this definition of the coherence, one may notice that, convex ordering is coherent with standard deviation, as a deviation risk measure \cite{article_Rockafellar_deviation_risk_meas}, based on result \ref{coherent_cvx_ord_varriance}, and with the classic Conditional Value at Risk CVaR assuming $X, Y$ (which is straightforward using the Rockafeller-Uryasev representation \cite{ROCKAFELLAR20021443}). 
\end{remark}

It is worth noting that, in the preceding proposition, claim \ref{comp_gauss_cvx_ord} admits a generalization for \emph{radial distributions}. This is stated in the following proposition and we refer the reader to \cite{jourdain:hal-02304190} for a proof.

\begin{Proposition}[Radial distributions (generalization)]
\label{gen_radial_distri}
Let $Z \in \mathbb{R}^q$ be a $q$-dimensional random vector having a radial distribution in the sense that
$$\forall \quad O \in \mathcal{O}\big(q, \mathbb{R}\big), \quad OZ \sim Z.$$
Let $A, B \in \mathbb{M}_{d, q}\big(\mathbb{R}\big)$. Then, we have the following equivalence
\begin{equation}
\label{equiv_order}
BB^\top - AA^\top \in \mathcal{S}^{+}\big(d, \mathbb{R}\big) \iff AZ \preceq_{cvx} BZ.
\end{equation}
\end{Proposition}

Then, to establish our main results, we will need a notion of convexity for matrix fields.

\begin{definition}
\label{def_cvx_matrix}
\begin{enumerate}[label=\roman*.]
	\item \label{pre_order} \textbf{Preorder}. Let $A, B \in \mathbb{M}_{d, q}\big(\mathbb{R}\big)$. We define the following preorder on $\mathbb{M}_{d, q}\big(\mathbb{R}\big)$
	\begin{equation}
	\label{special_mat_pre_order}
	A \preceq B \quad \text{if} \quad BB^\top - AA^\top \in \mathcal{S}^{+}\big(d, \mathbb{R}\big).
	\end{equation}

	\item \label{conv_mat_vol_def} $\bm{\preceq}$\textbf{-Convexity}. A matrix-valued function $\sigma : \mathbb{R}^d \to \mathbb{M}_{d, q}\big(\mathbb{R}\big)$ is $\bm{\preceq}$-convex if for every $x, y \in \mathbb{R}^d$ and every $\lambda \in [0,1]$, there exist $O_{\lambda, x}, O_{\lambda, y} \in \mathcal{O}\big(q, \mathbb{R} \big)$ such that
	\begin{equation}
	\label{cond_matrix_conv}
	\sigma\big(\lambda x + (1-\lambda)y  \big) \preceq \lambda \sigma(x)O_{\lambda, x} + (1-\lambda)\sigma(y)O_{\lambda, y}
	\end{equation}
 i.e.
$$\big(\lambda \sigma(x)O_{\lambda, x} + (1-\lambda)\sigma(y)O_{\lambda, y}\big)\big(\lambda \sigma(x)O_{\lambda, x} + (1-\lambda)\sigma(y)O_{\lambda, y}\big)^\top - \sigma \sigma^\top\big(\lambda x + (1-\lambda)y \big) \in \mathcal{S}^{+}\big(d, \mathbb{R}\big).$$
\end{enumerate}
\end{definition}

\begin{remark}
In particular, when $d=q=1$, by setting $O_{\lambda, x} = \sign\big(\sigma(x)\big)$ one checks that condition \eqref{cond_matrix_conv} is equivalent to the convexity of the real-valued function $|\sigma|$.
\end{remark}

\begin{remark}[$\preceq$-convexity: a quite general example]
\label{gen_ex_cnv_matrix}
Condition \eqref{cond_matrix_conv} for the $\preceq$-convexity may appear difficult to verify at first given a matrix-valued function. However, there exists a quite general class of matrix which satisfies this condition.

Let $\lambda_k : \mathbb{R} \to \mathbb{R}, k = 1, \ldots, q$ such that all $|\lambda_k|$ are convex. Set:
$$\sigma(x) := A \cdot \diag\big(\lambda_1(x), \ldots, \lambda_q(x)  \big) \cdot O , \hspace{0.2cm} A \in \mathbb{M}_{d, q}\big(\mathbb{R}\big), O \in \mathcal{O}\big(q, \mathbb{R}\big).$$
Then $\sigma$ is $\bm{\preceq}$-convex. We refer the reader to Lemma \ref{proof_lemme_ex_gen_mat_cvx} for a proof.
\end{remark}

\subsection{Main results}
For the rest of the paper, we assume that:

\vspace{0.3cm}
$\bm{(\mathcal{H})}$: $\overline{q}, \underline{Q}, \overline{Q} \in \mathbb{N}$, and $\overline{Q} - \underline{Q}$ is a multiple of $\overline{q}$. In particular, this implies that the supremum in the \emph{BDPP} is a maximum.

\vspace{0.3cm}
We establish two key results concerning the convexity of the swing price with respect to the underlying asset price and the functional \q{monotonicity} (also called \emph{domination}) with respect to the underlying volatility function. To this end, we consider the Brownian diffusion \eqref{general_diffusion_X} only at exercise dates which we sometimes denote by $\big(X_{t_k}^{[\sigma]}\big)_{0 \le k \le n}$ to emphasize the dependence of the process on the matrix-valued volatility function $\sigma : \mathbb{R}^d \to \mathbb{M}_{d, q}\big(\mathbb{R} \big)$.

Our swing pricing problem being a discrete-time problem, for $m \in \mathbb{N}^{*}$, we consider the Euler scheme, with step $\frac{T}{mn}$, of the diffusion \eqref{general_diffusion_X} at discretization dates $t_{\ell}^{(mn)} := \frac{\ell T}{m n}$ ($\ell \in \{0,\ldots,mn\}$):
\begin{equation}
    \label{euler_diffusion_X}
    \overline{X}^{x}_{t_{\ell+1}^{(mn)}} \sim \mathcal{E}_\ell^{(mn)}\Big(\overline{X}^{x}_{t_{\ell}^{(mn)}}, Z_{\ell+1}\Big), \quad \overline{X}_{0} = x \in \mathbb{R}^d,
\end{equation}
where
$$\mathcal{E}_\ell^{(mn)}(x, z) := x + h \cdot \kappa(t_{\ell}^{(mn)}) (x-\zeta) + \sqrt{h} \cdot \sigma_{t_{\ell}^{(mn)}}(x)z.$$
Note that the discretization dates, when $\ell = km$, coincide with the actual exercise dates of the swing contract i.e. $t_{km}^{(mn)} = t_k$ for $k=0,\ldots,n$. Besides, $(Z_{\ell})_\ell$ are \emph{i.i.d.} copies of $Z \sim \mathcal{N}(0, \mathbf{I}_q)$.

In the notation $\overline{X}^{x}_{t_{\ell}^{(mn)}}$, we opted not to highlight the step $h$ of the Euler scheme. We made this choice because, in this paper, only this Euler scheme with step $h$ will be used. Thus, there is no ambiguity and this choice allows to ease notations.

For a bounded or non-negative Borel function $f:\mathbb{R}^d \to \mathbb{R}$, we define the transition:
\begin{equation}
    \label{markov_transition}
    \forall x \in \mathbb{R}^d, \quad \mathrm{P}^{(mn)}_\ell(f)(x) := \mathbb{E}f\big(\mathcal{E}_\ell^{(mn)}(x, Z)\big), \quad 0 \le \ell \le mn-1
\end{equation}
as well as its composition i.e. for $i <j \in \{0,\ldots,mn-1\}$:
\begin{equation}
    \label{compo_transition_markov}
    \mathrm{P}^{(mn)}_{i:j} := \mathrm{P}^{(mn)}_i \circ \cdots \circ \mathrm{P}^{(mn)}_{j-1} \quad \text{and} \quad \mathrm{P}^{(mn)}_{\ell:\ell}(f) := f.
\end{equation}
Then, the \emph{BDPP} \eqref{eq_dp_swing_generic} reads
\begin{equation}
    \left\{
    \begin{array}{ll}
        v_{k}^{(m)}\big(x_{0:k}, Q\big) = \underset{q \in \mathbb{A}_c(t_k, Q)}{\sup} \hspace{0.1cm} \Big[\Psi_k(t_{k}, q, x_{0:k}) + \mathrm{P}^{(mn)}_{km:(k+1)m}\big(v_{k+1}^{(m)}(x_{0:k},\cdot, Q+q) \big)(x_k) \Big],\\
        \vspace{0.2cm}
        v_{n}^{(m)}(x_{0:n}, Q) = P_c(t_n, x_{0:n}, Q).
    \end{array}
\right.
\label{bdpp_with_markov_transition}
\end{equation}

Our aim is two-fold: we want to establish both the convexity of swing option prices and their \q{monotonicity} (comparison principle or \emph{domination}) with respect to the volatility function using a tractable (i.e. simulable) numerical approximation scheme for which the corresponding pricing problem shares the same properties. This is of course a crucial property for practitioners. Therefore, we will introduce some Euler schemes of the diffusion process and then rely on limit theorems to transfer the two above mentioned properties to the continuous-time dynamics.

We start with the following proposition for the propagation of convexity through the continuation value. We refer the reader to Appendix \ref{appendix_proof_some_prop} for a proof.

\begin{Proposition}
\label{useful_properties_cvx_ord}
Let $\ell \in \{0,\ldots, mn\}$ and $f : \mathbb{R}^d \to \mathbb{R}$ be a convex function with linear growth. We have the following results on $\mathrm{P}^{(mn)}_\ell$ (defined in \eqref{markov_transition}):
\begin{enumerate}[label=\roman*.]
\item \label{tf_convex} If $\sigma_{t_{\ell}^{(mn)}}$ is $\preceq$-convex then $\mathrm{P}^{(mn)}_\ell(f)(\cdot)$ is convex.

\item \label{croiss_mat} For $A \in \mathbb{M}_{d, q}(\mathbb{R})$, define:
\begin{equation}
    \label{op_transition_specific_sigma}
    \forall x \in \mathbb{R}^d, \quad \mathrm{P}^{(mn)}_{\ell, \sigma}(f)(x) := \mathbb{E}f\big(\mathcal{E}_{\ell, \sigma}^{(mn)}(x, Z)\big) := \mathbb{E}f\big(x + h \cdot \kappa(t_{\ell}^{(mn)}) (x-\zeta) + \sqrt{h} \cdot \sigma(x)Z \big).
\end{equation}
Then, for all $x \in \mathbb{R}^d$, $\mathbb{M}_{d, q}(\mathbb{R}) \ni \sigma \mapsto \mathrm{P}^{(mn)}_{\ell, \sigma}(f)(x)$ is $\preceq$-non-decreasing i.e. for any $A, B \in \mathbb{M}_{d, q}(\mathbb{R})$, 
$$A \preceq B \implies \mathrm{P}^{(mn)}_{\ell, A} (f)(x) \le \mathrm{P}^{(mn)}_{\ell, B}(f)(x).$$
\end{enumerate}
\end{Proposition}

\vspace{0.3cm}
Let us now consider the following assumptions for $c \in \{\emph{firm}, \emph{pen}\}$:

\vspace{0.3cm}
\noindent
$\bm{(\mathcal{H}_1^{c})}$: For all $k \in \{0, \ldots, n-1\}$, $q \in [0, \overline{q}]$ and $Q \in \mathcal{Q}_c(t_n)$, the payoff function
$$\big(\mathbb{R}^d\big)^{k+1} \ni x_{0:k} \mapsto \Psi_k\big(t_k, q, x_{0:k} \big)$$
and the penalty function,
$$ \big(\mathbb{R}^d\big)^{n+1} \ni x_{0:n} \mapsto P_c\big(t_n, x_{0:n}, Q)$$ 
are convex with at most linear growth for the $|\cdot|_{\sup}$-norm.

\vspace{0.4cm}
\noindent
$\bm{(\mathcal{H}_2^{[\sigma]})}$: For all $t \in [0, T]$, $\sigma_{t}\hspace{0.1cm} \text{is} \hspace{0.1cm} \bm{\preceq}$-convex.

\vspace{0.2cm}

\begin{theorem}[Convexity propagation: Euler scheme]
\label{prop_cvx_euler_cas_gen}
Let $c \in \{firm, pen\}$. Under assumptions $\bm{(\mathcal{H}_1^{c})}$ and $\bm{(\mathcal{H}_2^{[\sigma]})}$, for any $k \in \{0,\ldots,n\}$ and $Q \in \mathcal{Q}_c(t_k)$,
$$ (\mathbb{R}^d)^{k+1} \ni x_{0:k} \mapsto v_k^{(m)}\big(x_{0:k}, Q\big) \quad \text{is convex.}$$
\end{theorem}

\begin{proof}
We proceed by backward induction on $k$. Owing to Assumption $\bm{(\mathcal{H}_1^{c})}$, $v_{n}^{(m)}\big(\cdot, Q\big)$ is convex for any $Q \in \mathcal{Q}_c(t_n)$. Let us assume that the proposition holds for $k+1$. Let $x_{0:k}, y_{0:k} \in (\mathbb{R}^d)^{k+1}, \lambda \in [0,1]$, and define $\overline{x_k y_k}^{\lambda} := \lambda x_k + (1-\lambda) y_k, \overline{x_{0:k} y_{0:k}}^{\lambda} := \lambda x_{0:k} + (1-\lambda) y_{0:k}$. For any $Q \in \mathcal{Q}_c(t_k)$, we have,
$$v_k^{(m)}\big(\overline{x_{0:k} y_{0:k}}^{\lambda}, Q\big) = \underset{q \in \mathbb{A}_c(t_k, Q)}{\sup} \hspace{0.1cm} \Big[\Psi_k\big(t_k, q, \overline{x_{0:k} y_{0:k}}^{\lambda}\big) + \mathrm{P}^{(mn)}_{km:(k+1)m}\big(v_{k+1}^{(m)}(\overline{x_{0:k} y_{0:k}}^{\lambda}, \cdot, Q+q) \big)\big(\overline{x_k y_k}^{\lambda}\big) \Big].$$
By the induction assumption, $v_{k+1}^{(m)}(\overline{x_{0:k} y_{0:k}}^{\lambda},\cdot, Q)$ is convex for any $Q \in \mathcal{Q}_c(t_{k+1})$. The latter combined with assumption $\bm{(\mathcal{H}_2^{[\sigma]})}$, claim \ref{tf_convex} (in Proposition \ref{useful_properties_cvx_ord}), and a straightforward induction yield:
$$\mathbb{R}^d \ni x \mapsto \mathrm{P}^{(mn)}_{km:(k+1)m}\big(v_{k+1}^{(m)}(\overline{x_{0:k} y_{0:k}}^{\lambda},\cdot, Q+q) \big)(x) \quad \text{is convex}.$$
Then, using the convexity Assumption $\bm{(\mathcal{H}_1^{c})}$ of the payoff function, we get:
\begin{align*}
v_k^{(m)}\big(\overline{x_{0:k} y_{0:k}}^{\lambda} , Q\big) &\le \lambda \underset{q \in \mathbb{A}_c(t_k, Q)}{\sup} \hspace{0.1cm} \Big[\Psi_k\big(t_{k}, q, x_{0:k}\big) +  \mathrm{P}^{(mn)}_{km:(k+1)m}\big(v_{k+1}^{(m)}(x_{0:k}, \cdot, Q+q) \big)(x_k) \Big]\\
& \quad + (1 - \lambda) \underset{q \in \mathbb{A}_c(t_k, Q)}{\sup} \hspace{0.1cm} \Big[\Psi_k\big(t_{k}, q, y_{0:k}\big) +  \mathrm{P}^{(mn)}_{km:(k+1)m}\big(v_{k+1}^{(m)}(y_{0:k},\cdot, Q+q) \big)(y_k) \Big]\\
&= \lambda v_k^{(m)}\big(x_{0:k}, Q\big)  + (1 - \lambda) v_k^{(m)}\big(y_{0:k}, Q\big).
\end{align*}
This completes the proof.
\end{proof}

\begin{remark}[Path dependent payoff]
The convexity result has been established for the generic payoff functions $\Psi_k$ defined in \eqref{gen_payoff}. This includes possibly path dependent payoffs. Among others the most traded swing contracts are,

\begin{itemize}
	\item \textit{\textbf{Fixed strike swing contract}}. In this case, at each exercise date $t_k$, the holder of the swing contract receives, per unit of exercised volume, the difference between the forward price $F_{t_k}$ (e.g. gas delivery contract) and a fixed amount $K$ decided at the conclusion of the contract. That is,
	\begin{equation}
	\label{payoff_swing_fixed_strike}
	\Psi_k\big(t_k, q, x_{0:k}\big) := q \cdot \Big(f\big(t_k, x_{k} \big) - K \Big).
	\end{equation}
	which may appear as a \q{vanilla} swing contract following the equity terminology.

	\item \textit{\textbf{Indexed strike swing contract}}. This case is the same as the previous one except that the fixed amount $K$ is replaced by an average of past prices of the same commodity. Namely,
	\begin{equation}
	\label{payoff_swing_ind_strike}
	\Psi_k\big(t_k, q, x_{0:k}\big) := q \cdot \Big(f\big(t_k, x_{k} \big)  - \frac{1}{|I_k|} \sum_{i \in I_k}^{} f\big(t_{i}, x_{i} \big)  \Big),
	\end{equation}
	where for all $1 \le k \le n$, we have $I_k \subseteq \{0, \ldots, k-1\}$. One could also consider a case where the average is computed on past prices (still at exercise dates) of a different commodity asset (oil for instance). In this case, up to an enlargement of the filtration, one could achieve the same type of results proved in this paper.

    \item \textit{\textbf{Strip of call contract}}: We can consider a call version of the standard swing payoff\eqref{payoff_swing_fixed_strike}:
    \begin{equation}
        \label{call_payoff}
        \Psi_k\big(t_k, q, x_{0:k}\big) := q \cdot \Big(f\big(t_k, x_{k} \big) - K \Big)_{+}.
    \end{equation}
\end{itemize}
\end{remark}

The propagation of convexity for the actual value function \eqref{eq_dp_swing_generic} follows by sending $m$ to $+ \infty$. This follows from the next two propositions.

\begin{Proposition}
\label{lispc_func_val}
Assume the following properties.

\vspace{0.2cm}
\noindent
$\bm{(\mathcal{H}_3)}$ : For all $k \in \{0,\ldots, n-1\}$, $\Psi_k\big(t_k, q, \cdot\big)$ is Lipschitz continuous uniformly in $q \in [0, \overline{q}]$. Denote by $[\Psi_k]_{Lip}$ the Lipschitz coefficient.

\vspace{0.2cm}
\noindent
$\bm{(\mathcal{H}_4^c)}$ : For $c \in \{firm, pen\}$, $P_c\big(t_n, \cdot, Q)$ is Lipschitz continuous uniformly in $Q \in \mathcal{Q}_c(t_n)$. Denote by $[P_{c, n}]$ its Lipschitz coefficient. Note that $\big[P_{firm, n}\big]_{\text{Lip}} = 0$.

\vspace{0.2cm}
\noindent
Then, for all $k \in \{0,\ldots, n\}$, the swing value functions $v_k^{(m)}(\cdot, Q)$ are Lipschitz continuous uniformly in $Q \in \mathcal{Q}_c(t_k)$ with a Lipschitz coefficient $\big[v_k^{(m)}\big]_{\text{Lip}}$ satisfying:
$$\big[v_k^{(m)}\big]_{\text{Lip}} := \underset{Q \in \mathcal{Q}_c(t_k)}{\sup} \hspace{0.2cm} \underset{x_{0:k} \neq y_{0:k}}{\sup} \hspace{0.2cm} \frac{\big|v_k^{(m)}\big(x_{0:k}, Q\big) - v_k^{(m)}\big(y_{0:k}, Q\big)\big|}{|x_{0:k} - y_{0:k}|_{(k)}} \le  \sum_{i = k}^{n-1} C_{h, \beta, \sigma}^{mi} \big[\Psi_{i}\big]_{\text{Lip}} + C_{h, \beta, \sigma}^{mn} \big[P_{c,n} \big]_{\text{Lip}},$$
where $C_{h, \kappa, \sigma} := 1+hC_{\kappa, \sigma} := 1 + h\Big([\kappa]_{\infty} + \frac{[\sigma]_{Lip}^2}{2}\Big)$.
\end{Proposition}

\begin{remark}
\label{bounded_lip_coeff_vk}
\begin{itemize}
    \item Note that $[v_k^{(m)}]_{Lip}$ is bounded uniformly in $m$ by the constant $e^{t_i^{(n)}C_{\beta, \sigma}}$ where $t_i^{(n)} := \frac{Ti}{n}$ for $i \in \{0,\ldots, n\}$. Indeed, it suffices to notice that, using the classic inequality $1 + x \le e^x$ and the fact that $h = \frac{T}{mn}$, one has for all $i \in \{0,\ldots, n\}$:
    $$C_{h, \kappa, \sigma}^{mi} \le e^{mihC_{\kappa, \sigma}} = e^{t_i^{(n)}C_{\kappa, \sigma}}.$$

    \item Under Assumptions $\bm{(\mathcal{H}_3)}$ and using the Lipschitz property of $\sigma_{t_\ell^{(mn)}}$ (uniformly in $t_\ell^{(mn)}$), one may also show that the value function given by \eqref{eq_dp_swing_generic} is Lipschitz continuous. Indeed, by a straightforward backward induction, one may notice that:
    $$\big|v_k(x_{0:k}, Q) - v_k(y_{0:k}, Q)\big| \le [\Psi_k]_{Lip} \cdot \big|x_{0:k}-y_{0:k}\big|_{(k)} + [v_{k+1}]_{Lip} \cdot \Big\| X_{t_{k+1}}^{x_k, t_k} - X_{t_{k+1}}^{y_k, t_k} \Big\|_2.$$
   One completes the proof by noticing (see Theorem 7.10 in \cite{Pages2018}) that the flow $\mathbb{R}^d \ni x \mapsto X_{t_{k+1}}^{x, t_k}$ is Lipschitz continuous in $\mathbb{L}^2_{\mathbb{R}^d}(\mathbb{P})$ which is classic background.
\end{itemize}
\end{remark}

\begin{proof}[Proof of Proposition \ref{lispc_func_val}]
We rely on a backward induction on $k$. Assumption $\bm{(\mathcal{H}_4^c)}$ implies that the result holds for $k = n$. Assume now that the result holds for $k + 1$. Then, it follows from the triangle inequality that
\begin{align*}
\big|v_{k}^{(m)}(x_{0:k}, Q) &- v_{k}^{(m)}(y_{0:k}, Q) \big| \\
&\le \underset{q \in \mathbb{A}_c(t_k, Q)}{\sup} \hspace{0.1cm} \big|\Psi_{k}\big(t_{k}, q, x_{0:k}\big) - \Psi_{k}\big(t_{k}, q, y_{0:k}\big)\big| \\
& \quad + \underset{q \in \mathbb{A}_c(t_k, Q)}{\sup} \hspace{0.1cm} \Big|\mathrm{P}_{km:(k+1)m}^{(mn)}\big(v_{k+1}^{(m)}(x_{0:k},\cdot, Q + q)\big)(x_k) - \mathrm{P}_{km:(k+1)m}^{(mn)}\big(v_{k+1}^{(m)}(y_{0:k},\cdot, Q + q)\big)(y_k)\Big|.
\end{align*}

But note that, for any Lipschitz continuous function $f : \mathbb{R}^d \to \mathbb{R}$, one has for $\ell \in \{0,\ldots, mn-1\}$:
\begin{align*}
    \big|\mathrm{P}_\ell^{(mn)}(f)(x) -& \mathrm{P}_\ell^{(mn)}(f)(y) \big|\\
    &\le [f]_{Lip} \cdot \mathbb{E}\Big|x-y + h \cdot \kappa(t_\ell^{(mn)}) (x-y) + \sqrt{h}\big(\sigma_{t_\ell^{(mn)}}(x) - \sigma_{t_\ell^{(mn)}}(y)) Z\Big|\\
    &\le [f]_{Lip} \cdot \Big\|x-y + h \cdot \kappa(t_\ell^{(mn)}) (x-y) + \sqrt{h}\big(\sigma_{t_\ell^{(mn)}}(x) - \sigma_{t_\ell^{(mn)}}(y)) Z \Big\|_2\\
    &\le [f]_{Lip} \cdot |x-y| \cdot \Big(1 + h^2[\kappa]_{\infty}^2 + 2h[\kappa]_{\infty} + h[\sigma]_{Lip}^2  \Big)^{1/2}\\
    &\le [f]_{Lip} \cdot |x-y| \cdot \underbrace{\Bigg(1 + h\Big([\kappa]_{\infty} + \frac{[\sigma]_{Lip}^2}{2}\Big)\Bigg) }_{C_{h, \kappa, \sigma}}.
\end{align*}
This proves that $\mathbb{R}^d \ni x \mapsto \mathrm{P}(f)(x)$ is Lipschitz continuous. Thus, by a straightforward induction, one shows that for any $i < j \in \{0, \ldots, mn-1\}$, $\mathbb{R}^d \ni x \mapsto \mathrm{P}_{i:j}^{(mn)}(f)(x)$ is Lipschitz continuous and its Lipschitz coefficient satisfies:
$$\big[\mathrm{P}_{i:j}^{(mn)}(f)\big]_{Lip} \le [f]_{Lip} \cdot C_{h, \kappa, \sigma}^{j-i}.$$
Thus, we deduce that,
$$\big|v_{k}^{(m)}(x_{0:k}, Q) - v_{k}^{(m)}(y_{0:k}, Q) \big| \le \Big([\Psi_k]_{Lip} + [v_{k+1}^{(m)}]_{Lip} \cdot C_{h, \kappa, \sigma}^{m} \Big) \cdot \big|x_{0:k}-y_{0:k}\big|_{(k)}$$
so that,
$$[v_k^{(m)}]_{Lip} \le [\Psi_k]_{Lip} + [v_{k+1}^{(m)}]_{Lip} \cdot C_{h, \kappa, \sigma}^{m}.$$
Iterating this inequality yields the desired result.
\end{proof}

We then have the following limit result which allows the propagation of convexity to the actual swing value function.

\begin{Proposition}
\label{cvg_val_func_m}
Let $c \in \{firm, pen\}$ and consider Assumptions $\bm{(\mathcal{H}_3)}$, $\bm{(\mathcal{H}_4^{c})}$, . Then, for all $k \in \{0,\ldots, n\}$, and $Q \in \mathcal{Q}_c(t_k)$, and for every compact set $K \subset \big(\mathbb{R}^d\big)^{k+1}$, one has
$$ \lim_{m \to +\infty} \hspace{0.1cm} \underset{x_{0:k} \in K}{\sup} \hspace{0.1cm} \Big|v_k^{(m)}(x_{0:k}, Q)-v_k(x_{0:k}, Q)\Big| = 0.$$
\end{Proposition}

\begin{proof}
We proceed by a backward induction on $k$. The proposition clearly hold true for $k=n$ since both value functions coincide. Let us assume it holds for $k+1$. It follows from the classical inequality $\big|\underset{i \in I}{\sup} \hspace{0.1cm} a_i - \underset{i \in I}{\sup} \hspace{0.1cm} b_i \big| \le \underset{i \in I}{\sup} \hspace{0.1cm} \big|a_i-b_i\big|$, and the triangle inequality that:
    \begin{align*}
        \big|v_{k}^{(m)}(x_{0:k}, Q) - v_{k}(y_{0:k}, Q) \big| &\le \underset{Q_{k+1} \in \mathcal{Q}_c(t_{k+1})}{\sup} \hspace{0.1cm} \Big|\mathbb{E}v_{k+1}^{(m)}\big(x_{0:k}, \overline{X}^{x_k}_{t_{k+1}}, Q_{k+1}\big) - \mathbb{E}v_{k+1}^{(m)}\big(x_{0:k}, X^{x_k}_{t_{k+1}}, Q_{k+1}\big) \Big|\\
        &\quad + \underset{Q_{k+1} \in \mathcal{Q}_c(t_{k+1})}{\sup} \hspace{0.1cm} \Big|\mathbb{E}v_{k+1}^{(m)}\big(x_{0:k}, X^{x_k}_{t_{k+1}}, Q_{k+1}\big) - \mathbb{E}v_{k+1}\big(x_{0:k}, X^{x_k}_{t_{k+1}}, Q_{k+1}\big) \Big|.
    \end{align*}
Since Assumptions $\bm{(\mathcal{H}_3)}$ and $\bm{(\mathcal{H}_4^{c})}$ yield the Lipschitz property in Proposition \ref{lispc_func_val}, then:
\begin{align*}
    \underset{x_{0:k} \in K}{\sup} \hspace{0.1cm} \big|v_{k}^{(m)}(x_{0:k}, Q) - v_{k}(y_{0:k}, Q) \big| &\le \big[v_{k+1}^{(m)}\big]_{Lip} \cdot \underset{x_{0:k} \in K}{\sup} \hspace{0.1cm} \Big\|\overline{X}^{x_k}_{t_{k+1}} - X^{x_k}_{t_{k+1}}  \Big\|_1\\
    &\quad + \underset{Q_{k+1} \in \mathcal{Q}_c(t_{k+1}), x_{0:k} \in K}{\sup} \hspace{0.1cm} \mathbb{E}\Big|v_{k+1}^{(m)}\big(x_{0:k}, X^{x_k}_{t_{k+1}}, Q_{k+1}\big) - v_{k+1}\big(x_{0:k}, X^{x_k}_{t_{k+1}}, Q_{k+1}\big) \Big|.
\end{align*}
Let us deal with the r.h.s. sum term by term. For the first term, it converges toward 0 as $m \to +\infty$. This is a classic result on Euler scheme combined with the boundedness of $\big[v_{k+1}^{(m)}\big]_{Lip}$ in $m$ (see Remark \ref{bounded_lip_coeff_vk}). We recall that the (abusive) notation $\overline{X}^{x_k}_{t_{k+1}}$ of the Euler scheme hides the step $h$ which tends to 0 when $m \to +\infty$.

Owing to Assumption $\bm{\mathcal{H}}$, one can get rid of the supremum under $Q_{k+1} \in \mathcal{Q}_c(t_{k+1})$ in the second term. Besides, one has for any $R \ge 1$:
\begin{align*}
    \underset{x_{0:k} \in K}{\sup} \hspace{0.1cm} \mathbb{E}\Big|&v_{k+1}^{(m)}\big(x_{0:k}, X^{x_k}_{t_{k+1}}, Q_{k+1}\big) - v_{k+1}\big(x_{0:k}, X^{x_k}_{t_{k+1}}, Q_{k+1}\big) \Big|\\
    &\le \mathbb{E}\Bigg(\Big\|v_{k+1}^{(m)}\big(x_{0:k}, \cdot, Q_{k+1}\big) - v_{k+1}\big(x_{0:k}, \cdot, Q_{k+1}\big) \Big\|_{\mathcal{B}(0, R)}\Bigg)\\
    &\quad + \underset{x_{0:k} \in K}{\sup} \hspace{0.1cm} \mathbb{E}\Bigg(\Big|v_{k+1}^{(m)}\big(x_{0:k}, X^{x_k}_{t_{k+1}}, Q_{k+1}\big) - v_{k+1}\big(x_{0:k}, X^{x_k}_{t_{k+1}}, Q_{k+1}\big) \Big| \cdot \mathbf{1}_{\big|X^{x_k}_{t_{k+1}}\big| > R}\Bigg)\\
    &\le \mathbb{E}\Bigg(\Big\|v_{k+1}^{(m)}\big(x_{0:k}, \cdot, Q_{k+1}\big) - v_{k+1}\big(x_{0:k}, \cdot, Q_{k+1}\big) \Big\|_{\mathcal{B}(0, R)}\Bigg) + \frac{c_k}{R} \Big( 1 + \underset{x_{0:k} \in K}{\sup} \hspace{0.1cm} \big|x_{0:k}\big|_{(k)}^2 \Big),
\end{align*}
where the positive constant $c_k$ (only dependent on $k$) comes from the Lipschitz property of $v_{k+1}^{(m)}(\cdot, Q)$ (uniformly in $Q \in \mathcal{Q}_c(t_{k+1})$ and $m \in \mathbb{N}^{*}$), and that of $v_{k+1}(\cdot, Q)$ combined with the classic control of the moments of the flow $X_{t_k+1}^{x_k}$. Then, let us deal with the r.h.s. sum term by term.

For any $R \ge 1$, the first term converges toward 0 as $m \to +\infty$ using the induction assumption combined with the dominated convergence theorem which is justified by the Lipschitz property of $v_{k+1}^{(m)}(\cdot, Q)$ (uniformly in $Q \in \mathcal{Q}_c(t_{k+1})$ and $m \in \mathbb{N}^{*}$), and that of $v_{k+1}(\cdot, Q)$. Besides, the second term converges toward 0 as $R \to + \infty$. Thus taking successively the limit $m \to +\infty$, and $R \to +\infty$ ends the proof.
\end{proof}

The convexity result has some practical useful corollaries that we will discuss in the following remarks. Before that, it is worth noting that one can get rid of the drift in the affine setting.
\begin{remark}
\label{rq_null_drift}
    Define:
    $$\forall t \in [0,T], \quad \bm{\kappa}(t) :=  \int_{0}^{t} \kappa(s) \,ds \quad \text{and} \quad Y_t := e^{-\bm{\kappa}(t)} X_t + \zeta\big(1-e^{-\bm{\kappa}(t)} \big).$$
    Then, by Itô formula, one has:
    $$dY_t =  \widetilde{\sigma_t}\big(X_t\big)dW_t, \quad \widetilde{\sigma_t}(x) := e^{-\bm{\kappa}(t)} \cdot \sigma_t\Big(e^{\bm{\kappa}(t)} x  + \zeta\big(1-e^{\bm{\kappa}(t)}\big)\Big).$$
    Moreover, one checks that the $\preceq$-convexity of $\sigma_t$ is equivalent to that of $\widetilde{\sigma_t}$. Thus, in our proofs the affine drift may be canceled.
\end{remark}

\begin{remark}
The convexity property we have shown in the preceding proposition does not depend on volume constraints (assuming the space of constraints does not depend on the underlying price). Furthermore, we claim that, given a stochastic optimal control problem (either constrained or not) with its associated backward dynamic programming principle, where the set of constraints, if any, does not depend on the variable of interest, the same proof works.
\end{remark}

The convexity of the swing value function enables us to deduce our second main result concerning the functional monotonicity with respect to the matrix-valued volatility function $\sigma_t$. As discussed in Remark \ref{rq_null_drift}, one may simply consider an \emph{ARCH} martingale process instead of the general Euler scheme. This is what is done in what follows.

\begin{theorem}[Domination criterion]
\label{comparison}
Let $c \in \{firm, pen\}$. Consider assumption $\bm{(\mathcal{H}_1^{c})}$ and the following two \emph{ARCH} processes
\begin{equation}
\overline{X}_{t_{\ell + 1}^{(mn)}}^{[\vartheta]} = \overline{X}_{t_{\ell}^{(mn)}}^{[\vartheta]} + \vartheta_{t_{\ell}^{(mn)}}\Big(\overline{X}_{t_{\ell}^{(mn)}}^{[\vartheta]}\Big) Z_{\ell + 1}, \quad \overline{X}_{t_0}^{[\vartheta]} = x \in \mathbb{R}^d
\end{equation}
\begin{equation}
\overline{X}_{t_{\ell + 1}^{(mn)}}^{[\theta]} = \overline{X}_{t_{\ell}^{(mn)}}^{[\theta]} + \theta_{t_{\ell}^{(mn)}}\Big(\overline{X}_{t_{\ell}^{(mn)}}^{[\theta]}\Big) Z_{\ell + 1}, \quad \overline{X}_{t_0}^{[\theta]} = x \in \mathbb{R}^d,
\end{equation}
where $\theta_{t_{\ell}^{(mn)}}, \vartheta_{t_{\ell}^{(mn)}} : \mathbb{R}^d \to \mathbb{M}_{d, q}\big(\mathbb{R} \big)$ and $(Z_k)_k$ are \emph{i.i.d.} copies of $Z \in \mathcal{N}(0, \mathbf{I}_q)$. Assume that either assumption $\bm{(\mathcal{H}_2^{[\vartheta]})}$ or $\bm{(\mathcal{H}_2^{[\theta]})}$ holds as well as

\vspace{0.4cm}
\noindent
$\bm{(\mathcal{H}_3)}$ : For all $x \in \mathbb{R}^d, t \in [0, T]$, $\vartheta_t(x)\preceq \theta_t(x).$

\vspace{0.2cm}
\noindent
Then, for all $k \in \{0,\ldots, n\}$, $x_{0:k} \in \big(\mathbb{R}^d\big)^{k+1}$ and for all $Q_k \in \mathcal{Q}_c(t_k)$, one has
$$v_k^{(m), [\vartheta]}\big(x_{0:k}, Q_k \big) \le v_k^{(m), [\theta]}\big(x_{0:k}, Q_k \big),$$
where $v_k^{(m), [\vartheta]}$ and $v_k^{(m), [\theta]}$ are swing value functions defined by \eqref{bdpp_with_markov_transition} associated to Euler schemes $\overline{X}_{t_{\ell}^{(mn)}}^{[\vartheta]}$ and $\overline{X}_{t_{\ell}^{(mn)}}^{[\theta]}$.
\end{theorem}

\begin{proof}
We prove this proposition by a backward induction on $k$. The proposition holds for $k=n$ since $v_{n}^{[\vartheta]}\big(x_{0:n}, Q_{n} \big) = v_{n}^{[\theta]}\big(x_{0:n}, Q_{n} \big)$ for all $Q_n \in \mathcal{Q}_c(t_n)$. Assume now that it holds for $k+1$. For any $x_{0:k} \in (\mathbb{R}^d)^{k+1}$ and $Q_k \in \mathcal{Q}_c(t_k)$, we have:
$$v_{k}^{(m), [\vartheta]}\big(x_{0:k}, Q\big) = \underset{q \in \mathbb{A}_c(t_k, Q)}{\sup} \hspace{0.1cm} \Big[\Psi_k(t_{k}, q, x_{0:k}) + \mathrm{P}^{(mn), [\vartheta]}_{km:(k+1)m}\big(v_{k+1}^{(m), [\vartheta]}(x_{0:k},\cdot, Q+q) \big)(x_k) \Big],$$
where we denote by $\mathrm{P}^{(mn), [\vartheta]}_{km:(k+1)m}$ the composition of Markovian transition operators (see \eqref{compo_transition_markov}) associated to the volatility $\vartheta$. From Assumptions $\bm{(\mathcal{H}_1^{c})}$ and $\bm{(\mathcal{H}_2^{[\sigma]})}$ or $\bm{(\mathcal{H}_2^{[\theta]})}$, either $v_{k+1}^{(m), [\vartheta]}\big(\cdot, Q_k + q \big)$ or $v_{k+1}^{(m), [\vartheta]}\big(\cdot, Q_k + q \big)$ is a convex function as a consequence of Theorem \ref{prop_cvx_euler_cas_gen}. Then using claim \ref{croiss_mat} (in Proposition \ref{useful_properties_cvx_ord}) and Assumption $\bm{(\mathcal{H}_3)}$ yield
\begin{align*}
v_{k}^{(m), [\vartheta]}\big(x_{0:k}, Q_k \big) &\le \underset{q \in \mathbb{A}_c(t_k, Q)}{\sup} \hspace{0.1cm} \Big[\Psi_k(t_{k}, q, x_{0:k}) + \mathrm{P}^{(mn), [\theta]}_{km:(k+1)m}\big(v_{k}^{(m), [\vartheta]}(x_{0:k},\cdot, Q+q) \big)(x_k) \Big]\\
&\le \underset{q \in \mathbb{A}_c(t_k, Q)}{\sup} \hspace{0.1cm} \Big[\Psi_k(t_{k}, q, x_{0:k}) + \mathrm{P}^{(mn), [\theta]}_{km:(k+1)m}\big(v_{k}^{(m), [\theta]}(x_{0:k},\cdot, Q+q) \big)(x_k) \Big]\\
&= v_k^{(m), [\theta]}\big(x_{0:k}, Q_k \big),
\end{align*}
where we used the induction assumption in the second inequality. This completes the proof.
\end{proof}

The preceding \emph{domination criterion} also hold for the swing actual value function by sending $m$ to $+\infty$ which is allowed by Proposition \ref{cvg_val_func_m}. The \emph{domination criterion}, as a comparison principle, allows to compare swing option prices. To get convinced, let us take the following example.

\begin{example}[Domination with correlation]
    Let $\rho \in (-\frac{1}{q-1}, 1)$. For all $t \in [0, T]$, let us consider the following model for the volatility $\sigma_t$:
    $$\mathbb{R} \ni x \mapsto  \sigma_t(x) := \lambda_t(x)^\top L(\rho) \in \mathbb{M}_{1,q}(\mathbb{R}), \quad \lambda(x)^\top = \big(\lambda_{t,1}(x),\ldots,\lambda_{t,q}(x)\big) \in \mathbb{R}^q,$$
    where all functions $\mathbb{R} \ni x \mapsto \lambda_{t,i}(x)$ are non-negative and convex. $L(\rho)$ is the Cholesky decomposition of the correlation matrix $\Gamma(\rho) := \big[\rho + (1- \rho)\mathbf{1}_{i = j}]_{1 \le i, j \le q}$ which is definite positive. Then $\mathbb{R} \ni x \mapsto \sigma_t$ is $\preceq$-convex. Indeed, by simple algebra one has
    \begin{align*}
        &\big(\alpha \sigma_t(x) + (1-\alpha) \sigma_t(y)\big)\cdot \big(\alpha \sigma_t(x) + (1-\alpha) \sigma_t(y)\big)^\top\\
        &\quad - \sigma_t(\alpha x+ (1-\alpha)y)\sigma_t^\top(\alpha x+ (1-\alpha)y)\\
        &= \sum_{i, j = 1}^q \rho_{i, j} \Big[\alpha^2 \lambda_{t,i}(x)\lambda_{t,j}(x)+2\alpha(1-\alpha)\lambda_{t,i}(x)\lambda_{t,j}(y)+(1-\alpha)^2\lambda_{t,i}(y)\lambda_{t,j}(y) \\
        &\quad -\lambda_{t,i}(\alpha x + (1-\alpha)y)  \lambda_{t,j}(\alpha x + (1-\alpha)y)\Big] 
    \end{align*}
    and the r.h.s. is non-negative owing to the convexity of all non-negative functions $\lambda_{t,i}$. This implies the $\preceq$-convexity of volatility functions $\sigma_t$ as a straightforward application of Definition \eqref{cond_matrix_conv}. Besides, owing to Proposition \ref{mono_corr}, the monotonicity of $\sigma_t$ w.r.t. $\rho$ holds so that the domination criterion holds. This proves that, in this model, the swing price is increasing with the correlation parameter $\rho$.
\end{example}

\begin{remark}
In light of the proof of the domination criterion, one may notice that the terminal conditions, namely the penalty functions, do not need to be the same. Indeed, provided that:
$$v_n^{[\sigma]}(x_{0:n}, Q) = P_c^{[\sigma]}(t_n, x_{0:n}, Q), \quad v_n^{[\theta]}(x_{0:n}, Q) = P_c^{[\theta]}(t_n, x_{0:n}, Q)$$
with either penalty function $P_c^{[\sigma]}(t_n, \cdot, Q)$ or $P_c^{[\theta]}(t_n, \cdot, Q)$ being a convex function and such that $P_c^{[\sigma]}(t_n, \cdot, Q) \le P_c^{[\theta]}(t_n, \cdot, Q)$. Then, under either assumption $\bm{(\mathcal{H}_2^{[\sigma]})}$ (in the first case) or $\bm{(\mathcal{H}_2^{[\theta]})}$ (in the second case), the same result holds true.
\end{remark}

\section{One dimensional case: some refinements}
\label{refinements_oneD}
The \emph{domination criterion} and the propagation of convexity proved above rely on the convexity assumption $\bm{(\mathcal{H}_2^{[\sigma]})}$ for the matrix-valued volatility function $\sigma_t$. In full generality, this assumption cannot be reasonably relaxed. This section focuses on the one-dimensional setting i.e., $d=q=1$, where we prove that, in this specific case, it is possible to get rid of the convexity assumption.

From now on, in this section, we set $d=q=1$ and consider \q{vanilla} payoff functions, at time $t_k$, which only depend on the price at time $t_k$, i.e. of the form:
\begin{align}
\label{payoff_func_1d}
\Psi_k \colon [0, T] \times \mathbb{R}_{+} \times \mathbb{R}  &\to \mathbb{R} \nonumber \\
	(t_k, q_k, x_{k}) &\mapsto \Psi_k\big(t_k, q_k, x_{k}\big), \quad k \in \{0,\ldots,n-1\}.
\end{align}

Our aim is to relax the convexity assumption of the volatility function $\sigma_t$ with a semi-convexity assumption.

\vspace{0.2cm}
\noindent
\underline{Outline of the proof}: We first proceed as in \cite{jourdain2023convex} by truncating the Gaussian noise $Z$. That is, $Z$ is replaced by
\begin{equation}
    \label{Z_trunc}
    \tilde{Z}^h := Z \cdot \mathrm{1}_{\{|Z| \le s_h\}}
\end{equation}
for a positive threshold $s_h$ such that $s_h \to +\infty$ when $h \to 0$. We then show in Proposition \ref{approx_dim1_conv} that, for an appropriate choice of $s_h$, the truncated Euler $\mathcal{E}_\ell^{(mn)}(x, \tilde{Z}^h)$ propagates convexity i.e. for any convex function $f:\mathbb{R} \to \mathbb{R}$, the function $\mathbb{R} \ni x \mapsto \mathbb{E}f\big(\mathcal{E}_\ell^{(mn)}(x, \tilde{Z}^h)\big)$ is convex. This allows us to prove Proposition \ref{cnv_dim1_trunc} which states that convexity actually propagates through the \emph{truncated BDPP} (see \eqref{eq_dp_swing_trunc}) i.e. the \emph{BDPP} \eqref{bdpp_with_markov_transition} where the actual white noise $Z$ is replaced by its truncated version $\tilde{Z}^h$. To show that the preceding result still holds true for the \emph{BDPP} of interest given by Equation \eqref{eq_dp_swing_generic}, we send $h \to 0$.

Let us start by showing that the truncated Euler scheme $\mathcal{E}_\ell^{(mn)}(x, \tilde{Z}^h)$ propagates convexity.

\begin{Proposition}
\label{approx_dim1_conv}
Define:
$$\mathcal{E}^h\big(x, \tilde{Z}^h \big) := x + h \beta(x) + \sqrt{h} \sigma(x) \tilde{Z}^h.$$
Assume that the volatility function $\sigma : \mathbb{R} \to \mathbb{R}_{+}$ is Lipschitz continuous and semi-convex i.e.,
\begin{equation}
\label{semi_cvx_hyp}
a_{\sigma} := \inf \big\{a \ge 0 :  \mathbb{R} \ni x \mapsto \sigma^2(x) + ax^2 \hspace{0.2cm} \text{is convex} \big\} < + \infty.
\end{equation}
Furthermore, assume that the drift function $\beta : \mathbb{R} \to \mathbb{R}$ is convex and such that,
\begin{equation}
\label{cond_drit_dim1_refin}
c_{\beta} := \inf \big\{c \ge 0 : \mathbb{R} \ni x \mapsto \beta(x) + cx \hspace{0.2cm} \text{is non-decreasing}\big\} < + \infty.
\end{equation}
Let $h \in \big(0, \frac{1}{2 c_{\beta}}\big)$ (with the convention $\frac{1}{0} = + \infty$) and set $s_h = \frac{\lambda}{\sqrt{h \cdot \big([\sigma]_{Lip}^2 + a_{\sigma}  \big)}}$ for some $\lambda \in \big(0, \frac{1}{2 + \sqrt{2}} \big)$. Then, the following propositions hold true:
\begin{countlist}[label={(\alph*)}]{otherlist1}
	\item \label{euller_croissance} The random function $\mathbb{R} \ni x \mapsto \mathcal{E}^h\big(x, \tilde{Z}^h \big)$ is non-decreasing when $\mathbb{L}^1_{\mathbb{R}}(\mathbb{P})$ is equipped with the stochastic order.
	
	\item \label{euller_inc_cvx} The random function $\mathbb{R} \ni x \mapsto \mathcal{E}^h\big(x, \tilde{Z}^h \big)$ is non-increasing when $\mathbb{L}^1_{\mathbb{R}}(\mathbb{P})$ is equipped with the increasing convex order.
	
	\item \label{affine_cond_drift} If $\beta$ is affine, the random function $\mathbb{R} \ni x \mapsto \mathcal{E}^h\big(x, \tilde{Z}^h \big)$ is convex when $\mathbb{L}^1_{\mathbb{R}}(\mathbb{P})$ is equipped with the convex order.
\end{countlist}
\end{Proposition}

\begin{proof}[Proof of Proposition \ref{approx_dim1_conv}]
Our proof is divided into five steps and starts with some preliminary results. Note that to prove claim \ref{euller_inc_cvx} (resp. claim \ref{affine_cond_drift}), we need to prove that for any function $f : \mathbb{R} \to \mathbb{R}$ which is convex and non-decreasing (resp. convex) that $x \mapsto \mathbb{E}\big[f\big(\mathcal{E}^h(x, \tilde{Z}^h) \big]$ is non-decreasing (resp. convex). In our proof scheme, until \emph{Step 5}, we will prove these results for $f$ being twice continuously differentiable, convex and non-decreasing (resp. convex). At \emph{Step 5}, the $\mathcal{C}^2$ regularity will be relaxed.

Let us start with the first step that jumbles some useful results which will be used in the subsequent analysis. We assume that the volatility function is not constant. The constant case is not of interest as it implies the volatility function is convex and this has already been handled.

\vspace{0.2cm}
\item (Step 1). (\textit{Preliminaries}). Note that the volatility function $\sigma$ is not constant so that $[\sigma]_{Lip} > 0$ and $s_h < + \infty$. Let $\rho$ be a $\mathcal{C}^{\infty}$ probability density on the real line with compact support such that,
\begin{equation}
\label{hyp_int_convol}
\int_{\mathbb{R}}^{} u \rho(u) \, \mathrm{d}u = 0  \quad \text{and} \quad \int_{\mathbb{R}}^{} u^2 \rho(u) \, \mathrm{d}u = 1.
\end{equation}
$\rho$ is then associated to its sequence of modifiers $\rho_p(x) = p \cdot \rho(p x)$ for all $p \in \mathbb{N}^{*}$. We also set, for every $p \in \mathbb{N}^{*}$,
$$\sigma_p(x) := \sqrt{\frac{1}{p} + \rho_p \ast \sigma^2(x)} \quad \text{and} \quad \beta_p(x) := \rho_p \ast \beta(x).$$
The continuity of the convex function $\beta$ implies that $\beta_p$ converges pointwise to $\beta$ as $p \to + \infty$. Likewise, $\sigma_p$ converges pointwise to $\sigma$. Thus, to prove the three results of the proposition when $\sigma$ is non-constant, we start by proving them when replacing the random function $\mathcal{E}^h(x, \tilde{Z}^h)$ with, for each $p \in \mathbb{N}^{*}$,
$$ \mathbb{R} \ni x \mapsto \mathcal{E}^h_p(x, \tilde{Z}^h) := x + h \beta_p(x) + \sqrt{h} \sigma_p(x) \tilde{Z}^h.$$

It follows from triangle inequality and then Cauchy Schwartz's one that, for every $x,y \in \mathbb{R}$:
\begin{align*}
\big|\sigma_p^2(x) - \sigma_p^2(y) \big| &\le \int_{\mathbb{R}}^{} \big|\sigma(x-z) - \sigma(y-z) \big| \cdot \big(\sigma(x-z) + \sigma(y-z) \big) \rho_p(z) \, \mathrm{d}z\\
&\le [\sigma]_{Lip} \cdot |x-y| \cdot \Bigg(\int_{\mathbb{R}}^{} \sigma(x-z) \rho_p(z) \, \mathrm{d}z + \int_{\mathbb{R}}^{} \sigma(y-z) \rho_p(z) \, \mathrm{d}z \Bigg)\\
&\le [\sigma]_{Lip} \cdot |x-y| \cdot \Big(\int_{\mathbb{R}}^{} \rho_p(z) \, \mathrm{d}z \Big)^{1/2} \cdot \Bigg(\Big(\int_{\mathbb{R}}^{} \sigma^2(x-z) \rho_p(z) \, \mathrm{d}z\Big)^{1/2}\\
&\quad + \Big(\int_{\mathbb{R}}^{} \sigma^2(y-z) \rho_p(z) \, \mathrm{d}z\Big)^{1/2} \Bigg)\\
&\le [\sigma]_{Lip} |x-y| \big(\sigma_p(x) + \sigma_p(y) \big),
\end{align*}
where in the last inequality we use the fact that $\int_{\mathbb{R}}^{} \rho_p(z) \, \mathrm{d}z = 1$, since $\rho$ is a probability density and, the definition of $\sigma_p$. Thus, $\sigma_p$ is Lipschitz with $[\sigma_p]_{Lip} \le [\sigma]_{Lip}$. Moreover, since the Lipschitz continuous function $\sigma$ has at most affine growth, the non-negative function $\rho_p \ast \sigma^2$ is differentiable and so is $ \sigma_p$. Thus:
\begin{equation}
 \label{lipschtz_bound_sig}
 |\sigma^{'}_p| \le [\sigma]_{Lip}.
\end{equation}
Besides, we know that, $x \mapsto \sigma^2(x) + a_{\sigma} x^2$ is convex since the infimum defining $a_\sigma$ holds as a minimum. Hence, its convolution by $\rho_p$ is convex too and one classically checks that it is infinitely differentiable. On the other hand, using \eqref{hyp_int_convol}, one has
$$\int_{\mathbb{R}}^{} \Big(\sigma(x-y)^2 + a_{\sigma} \cdot (x-y)^2  \Big) \rho_p(y) \, \mathrm{d}y = \sigma_p^2(x) + a_{\sigma} x^2 + \frac{a_{\sigma}}{p^2} - \frac{1}{p}$$
so that $x \mapsto \sigma_p^2(x) + a_{\sigma} x^2$ is convex with $\big(\sigma_p^2)^{''} \ge -2a_{\sigma}$. We also note that $\beta_p$ is infinitely differentiable, convex and such that $x \mapsto \beta_p(x) + c_{\beta} x$ is non-decreasing.

\vspace{0.2cm}
\item (Step 2). (\textit{Differentiation}). This step aims at formally differentiating the function $ \mathbb{R} \ni x \mapsto \mathbb{E}\big[f\big(\mathcal{E}_p^h(x, \tilde{Z}^h)\big]$ which will be the key to prove claims \ref{euller_inc_cvx} and \ref{affine_cond_drift}.

\vspace{0.1cm}
Note that if $f : \mathbb{R} \mapsto \mathbb{R}$ is convex and twice continuously differentiable, one shows that the function $\mathbb{R} \ni x \mapsto \mathbb{E}\big[f\big(\mathcal{E}_p^h(x, \tilde{Z}^h)  \big] $ is twice continuously differentiable since the random variable $\tilde{Z}^h$ is bounded by $s_h < + \infty$. Its first two partial derivatives given by:
\begin{equation}
 \label{first_deriv_cont}
 \partial_x \mathbb{E}\big[f\big(\mathcal{E}_p^h(x, \tilde{Z}^h)  \big] = \mathbb{E}\Big[f^{'}\big(\mathcal{E}_p^h(x, \tilde{Z}^h  \big)\big) \cdot \big(1 + \sqrt{h} \sigma_p^{'}(x) \tilde{Z}^h + h \beta_p^{'}(x) \big)  \Big]
\end{equation}
and
\begin{equation}
 \label{snd_deriv_cont}
 \begin{split}
      \partial_{xx} \mathbb{E}\big[f\big(\mathcal{E}_p^h(x, \tilde{Z}^h)  \big] &= \mathbb{E}\Big[f^{'}\big(\mathcal{E}_p^h(x, \tilde{Z}^h  \big)\big)  \cdot \big(\sqrt{h} \sigma^{''}(x) \tilde{Z}^h + h \beta_p^{''}(x) \big)  \Big]\\
      &\quad + \mathbb{E}\Big[f^{''}\big(\mathcal{E}_p^h(x, \tilde{Z}^h  \big)\big) \cdot \big(1 + \sqrt{h} \sigma_p^{'}(x) \tilde{Z}^h + h\beta_p^{'}(x)\big)^2  \Big].
 \end{split}
\end{equation}

\vspace{0.2cm}
\noindent
\item (Step 3). (\textit{Claims \ref{euller_croissance} and \ref{euller_inc_cvx} for $\mathcal{E}_p^h(\cdot, \Tilde{Z}^h)$}). 

\vspace{0.1cm}
We start by proving \ref{euller_croissance}. When $h \le \frac{1}{2 c_{\beta}}$, since $x \mapsto \beta_p(x) + c_{\beta} x$ is non-decreasing, one has $h \beta_p^{'}(x) \ge -c_{\beta}h \ge -\frac{1}{2}$ and, by definition of the threshold $s_h$ and using \eqref{lipschtz_bound_sig}, one has
$$\partial_x \mathcal{E}_p^h(x, \tilde{Z}^h) = 1 + \sqrt{h}\sigma_p^{'}(x)\tilde{Z}^h + h \beta_p^{'}(x) \ge 1 - \sqrt{h} |\sigma_p^{'}(x)| s_h - \frac{1}{2} \ge 1 - \lambda - \frac{1}{2} > 0$$
since $\lambda < \frac{1}{2 + \sqrt{2}} < 1/2$. Therefore, $x \mapsto  \mathcal{E}_p^h(x, \tilde{Z}^h)$ is non-decreasing with respect to the non-decreasing stochastic ordering. Thus letting $p \to + \infty$ implies \ref{euller_croissance}. It remains to prove claim \ref{euller_inc_cvx}.

\vspace{0.1cm}
Having in mind that $|\sigma^{'}| \le [\sigma]_{Lip}$, one has
\begin{equation}
\label{eq_help_dim_1_refin_1}
1 + \sqrt{h} \sigma_p^{'}(x)\tilde{Z}^h + h\beta_p^{'}(x) \ge \frac{1}{2} - \sqrt{h} [\sigma]_{Lip} \cdot s_h \ge \frac{1}{2} - \lambda >0.
\end{equation}
Thus, if $f$ is convex, non-decreasing and twice continuously differentiable, then the partial derivative in \eqref{first_deriv_cont} is non-negative. Hence, the random function $\mathcal{E}^h_p\big(\cdot, \tilde{Z}^h\big)$ is non-decreasing for the increasing convex order. This partially proves \ref{euller_inc_cvx} since, as already mentioned, we need to prove the same result, but for the random function $\mathcal{E}^h\big(\cdot, \tilde{Z}^h\big)$ and without assuming that $f$ is twice differentiable.

\vspace{0.2cm}
\noindent
\item (Step 4). (\textit{Claim \ref{affine_cond_drift} for $\mathcal{E}_p^h(\cdot, \Tilde{Z}^h)$}). 

\vspace{0.1cm}
It follows from an integration by parts similar to Stein's Lemma that for a twice continuously differentiable function $f : \mathbb{R} \to \mathbb{R}$, one has
	\begin{align*}
	\mathbb{E}\Big[f^{'}\big(\mathcal{E}_p^h(x, \tilde{Z}^h \big)\big) \tilde{Z}^h  \Big] &= \int_{-s_h}^{s_h} f^{'}\big(\mathcal{E}_p^h(x, z)  \big) z e^{-\frac{z^2}{2}} \, \frac{\mathrm{d}z}{\sqrt{2\pi}}\\
	&= \int_{-s_h}^{s_h} f^{''}\big(\mathcal{E}_p^h(x, z)\big) \sqrt{h} \sigma_p(x)   \Big(e^{-\frac{z^2}{2}} - e^{-\frac{s_h^2}{2}}\Big) \, \frac{\mathrm{d}z}{\sqrt{2\pi}}\\
	&= \mathbb{E}\Bigg[f^{''}\big(\mathcal{E}_p^h(x, \tilde{Z}^h) \big) \sqrt{h} \sigma_p(x) \mathrm{1}_{\{\tilde{Z}^h \neq 0\}} \underbrace{\Big(1 - e^{- \frac{s_h^2 - (\tilde{Z}^h)^2}{2}}  \Big)}_{\ge 0}  \Bigg].
	\end{align*}
Plugging this equality in \eqref{snd_deriv_cont} yields,
\begin{equation}
\label{der_snd_transfo}
    \begin{split}
        \partial_{xx} \mathbb{E}\big[f\big(\mathcal{E}_p^h (x, \tilde{Z}^h)\big] &= \mathbb{E}\left[f^{'}\big(\mathcal{E}_p^h(x, \tilde{Z}^h) \big) h \beta_p^{''}(x) \right]\\
& \quad + \mathbb{E}\Bigg[f^{''}\big(\mathcal{E}_p^h(x, \tilde{Z}^h) \Big( \big(1 + \sqrt{h} \sigma_p^{'}(x) \tilde{Z}^h + h \beta_p^{'}(x) \big)^2 \\
&\quad + h \mathrm{1}_{\{\tilde{Z}^h \neq 0\}} \big(1 - e^{- \frac{s_h^2 - (\tilde{Z}^h)^2}{2}}  \big) \sigma_p \sigma_p^{''}(x)  \Big) \Bigg].
    \end{split}
\end{equation}
When $f$ is non-decreasing, then the first expectation in the right-hand side is non-negative by the convexity of $\beta_p$. It is still non-negative disregarding the monotonicity of $f$ when $\beta$ is affine since $\beta_p^{''}$ vanishes. Let us handle the second expectation. Using the identity $\sigma_p \sigma_p^{''}(x) = \frac{1}{2}\big(\sigma_p^2 \big)^{''} - \big(\sigma_p^{'} \big)^2$, the definition of $a_{\sigma}$ and the elementary inequality $1 - e^{-u} \le u$, one has
\begin{align}
\label{eq_help_dim_1_refin_2}
h \Big(1 - e^{-\frac{s_h^2 - (\tilde{Z}^h)^2}{2}} \Big) \sigma_p \sigma_p^{''}(x) &= \frac{h}{2} \Big(1 - e^{-\frac{s_h^2 - (\tilde{Z}^h)^2}{2}} \Big) \big((\sigma_p^2)^{''} + 2 a_{\sigma} \big)  - h \Big(1 - e^{-\frac{s_h^2 - (\tilde{Z}^h)^2}{2}} \Big) \big((\sigma_p^{'})^2 + a_{\sigma}  \big) \nonumber\\
&\ge -h \Big(1 - e^{-\frac{s_h^2 - (\tilde{Z}^h)^2}{2}} \Big) \big((\sigma_p^{'})^2 + a_{\sigma}  \big) \nonumber \\
&\ge -h \frac{s_h^2 - (\tilde{Z}^h)^2}{2} \big([\sigma]_{Lip}^2 + a_{\sigma} \big) \nonumber \\
&\ge -\frac{h s_h^2}{2} \big([\sigma]_{Lip}^2 + a_{\sigma}  \big) = - \frac{\lambda^2}{2}.
\end{align}
Inequalities \eqref{eq_help_dim_1_refin_1} and \eqref{eq_help_dim_1_refin_2} imply that, on the event $\{\tilde{Z}^h \neq 0 \}$,
$$\big(1 + \sqrt{h} \sigma_p^{'}(x) \tilde{Z}^h + h \beta_p^{'}(x) \big)^2 + h \mathrm{1}_{\{\tilde{Z}^h \neq 0\}} \Big(1 - e^{-\frac{s_h^2 - (\tilde{Z}^h)^2}{2}}  \Big) \sigma_p \sigma_p{''}(x) \ge \big(\frac{1}{2} - \lambda \big)^2 - \frac{\lambda^2}{2} > 0$$
since $\lambda < 1 - 1/\sqrt{2} = \frac{1}{2 + \sqrt{2}}$. Note that the latter expression is also positive on $\{\tilde{Z}^h = 0\}$. As $f^{''}$ is non-negative when $f$ is convex, we deduce that the second expectation in the right-hand side of \eqref{der_snd_transfo} is non-negative. Hence, when $f$ is moreover non-decreasing or $\beta$ is affine (so that $\beta_p$ is affine too), we get:
 $$\partial_{xx} \mathbb{E}\big[f\big(\mathcal{E}_p^h(x, \tilde{Z}^h)\big] \ge 0,$$
which implies that the random function $\mathcal{E}_p^h\big(\cdot, \tilde{Z}^h\big)$ is convex for the convex ordering.

\vspace{0.2cm}
\noindent
\item (Step 5). (\textit{Claims \ref{euller_inc_cvx} and \ref{affine_cond_drift}: General form}).

\vspace{0.1cm}
Let $f: \mathbb{R} \to \mathbb{R}$ be a twice continuously differentiable function. As already mentioned $\beta_p \to \beta$ and $\sigma_p \to \sigma$ pointwise and $\mathcal{E}_p^h(x, \tilde{Z}^h)$ is a bounded random variable when $h$ is fixed. When $h \in (0, \frac{1}{2 c_{\beta}})$, if $f : \mathbb{R} \to \mathbb{R}$ is $\mathcal{C}^2$, non-decreasing and convex, one has almost surely $f\big(\mathcal{E}_p^h(x, \tilde{Z}^h)\big) \to f\big(\mathcal{E}^h(x, \tilde{Z}^h)\big)$ as $p \to +\infty$ and
\begin{equation}
    \label{cvg_E_n_to_E}
    \mathbb{E} f\big(\mathcal{E}^h(x, \tilde{Z}^h) \big ) = \lim_{p\to\infty} \mathbb{E}f\big(\mathcal{E}_p^h(x, \tilde{Z}^h) \big)
\end{equation}
since $f$ is locally bounded. The same holds true for regular convex order (with $f$ convex and $\mathcal{C}^2$) when $\beta$ is affine. Thus, owing to \eqref{cvg_E_n_to_E}, claims \ref{euller_inc_cvx} and \ref{affine_cond_drift} previously established for the random function $\mathcal{E}_p^h(\cdot, \tilde{Z}^h)$ also hold for the random function $\mathcal{E}^h(\cdot, \tilde{Z}^h)$ but still assuming $f$ is twice continuously differentiable.

\vspace{0.1cm}
In order to relax the regularity of $f$, we proceed as previously by considering $f_p := f \ast \rho_p$.  Assume $f : \mathbb{R} \to \mathbb{R}$ is non-decreasing and convex. Then $f$ is continuous, $f_p$ is well-defined, $\mathcal{C}^{\infty}$, convex and $f_p \to f$ pointwise. Moreover, $\underset{p}{\sup} \hspace{0.1cm} |f_p|$ is bounded on every compact interval $[-A, A] \subset \mathbb{R}$ ($A > 0$), by $\underset{|y| \le A, u \in \supp(\rho)}{\sup} \hspace{0.1cm} |f(y-u)| < + \infty$. As the random variables $\mathcal{E}^h(x, \tilde{Z}^h), x \in [-A,A]$, have values in a fixed compact, it follows by the dominated convergence theorem and result \eqref{cvg_E_n_to_E} that, for every $x \in \mathbb{R}$,
\begin{equation}
\label{passage_lim_dim1_refin}
\mathbb{E} f\big(\mathcal{E}^h(x, \tilde{Z}^h) \big ) =  \lim_{j\to\infty} \mathbb{E}f_j\big(\mathcal{E}^h(x, \tilde{Z}^h) \big)=  \lim_{j\to\infty} \lim_{p\to\infty} \mathbb{E}f_j\big(\mathcal{E}_p^h(x, \tilde{Z}^h) \big),
\end{equation}
hence $\mathbb{R} \ni x \mapsto \mathbb{E} f\big(\mathcal{E}^h(x, \tilde{Z}^h) \big )$ is non-decreasing since we have previously shown that for all $j$, $\mathbb{R} \ni x \mapsto \mathbb{E} f_j\big(\mathcal{E}_p^h(x, \tilde{Z}^h) \big )$ is non-decreasing. This completes the proof of \ref{euller_inc_cvx}. To establish convex ordering, we can restrict to Lipschitz convex functions owing to Lemma \ref{rq_carac_ord_cvx}. Assume $f$ to be Lipschitz and convex. Then all the functions $f_p$ defined as above are well-defined and uniformly Lipschitz since $[f_p]_{Lip} \le [f]_{Lip}$. Still using that $\mathcal{E}_p^h(x, \tilde{Z}^h)$ is a bounded random variable, we derive that \eqref{passage_lim_dim1_refin} still holds true for every $x\in \mathbb{R}$ owing to the dominated convergence theorem. This allows to transfer convex ordering. Thus claim \ref{affine_cond_drift} is completely proved.
\end{proof}

We now consider the \emph{truncated BDPP}, replacing the random noise $Z$ by its truncated version $\tilde{Z}^h$ in the \emph{BDPP} \eqref{bdpp_with_markov_transition}. That is, we consider:
\begin{equation}
    \left\{
    \begin{array}{ll}
        \tilde{v}_k^{(m)}\big(x, Q \big) = \underset{q \in \mathbb{A}_c(t_k, Q)}{\sup} \hspace{0.1cm} \Big[ \Psi_k(t_k, q, x) + \tilde{\mathrm{P}}^{(mn)}_{km:(k+1)m}\big(\tilde{v}_{k+1}^{(m)}(\cdot, Q+q) \big)(x)\Big],\\
        \vspace{0.2cm}
        \tilde{v}_{n}^{(m)}(x, Q) = P_c(t_n, x, Q),
    \end{array}
\right.
\label{eq_dp_swing_trunc}
\end{equation}
where for $i < j \in \{0,\ldots,mn-1\}$ (still with $h = \frac{T}{mn}$), we used the following notations:
\begin{equation}
    \label{trunc_transition_def}
    \tilde{\mathrm{P}}^{(mn)}_{i:j} := \tilde{\mathrm{P}}^{(mn)}_i \circ \cdots \circ \tilde{\mathrm{P}}^{(mn)}_{j-1} \quad \text{with} \quad \tilde{\mathrm{P}}^{(mn)}_\ell(f)(x) := \mathbb{E}f\big(\mathcal{E}_\ell^{(mn)}(x, \tilde{Z}^h)\big) \quad \text{for} \quad 0 \le \ell  \le mn -1.
\end{equation}

The following proposition shows that the convexity propagates through this \emph{truncated BDPP}.

\begin{Proposition}[Convexity propagation: truncated BDPP]
\label{cnv_dim1_trunc}
Let $c \in \{firm, pen\}$. Under assumption $\bm{(\mathcal{H}_1^{c})}$ and if, in addition, assumptions of Proposition \ref{approx_dim1_conv} hold true, then for any $k \in \{0,\ldots,n\}$ and $Q \in \mathcal{Q}_c(t_k)$,
$$ \mathbb{R} \ni x \mapsto \tilde{v}_k^{(m)}\big(x, Q\big) \quad \text{is convex.}$$
\end{Proposition}

\begin{proof}
We proceed by backward induction on $k$. Owing to Assumption $\bm{(\mathcal{H}_1^{c})}$, $\tilde{v}_{n}^{(m)}\big(\cdot, Q\big)$ is convex for any $Q \in \mathcal{Q}_c(t_n)$. Let us assume that the proposition holds for $k+1$. Let $x_{k}, y_{k} \in \mathbb{R}$ and $\lambda \in [0,1]$ define $\overline{x_ky_k}^{\lambda} := \lambda x + (1-\lambda)y$. For any $Q \in \mathcal{Q}_c(t_k)$, we have,
$$\tilde{v}_k^{(m)}\big(\overline{x_ky_k}^{\lambda} , Q\big) = \underset{q \in \mathbb{A}_c(t_k, Q)}{\sup} \hspace{0.1cm} \Big[\Psi_k\big(t_k, q, \overline{x_ky_k}^{\lambda}\big) + \tilde{\mathrm{P}}^{(mn)}_{km:(k+1)m}\big(\tilde{v}_{k+1}^{(m)}(\cdot, Q+q) \big)\big(\overline{x_ky_k}^{\lambda}\big) \Big].$$
Since by the induction assumption, $\tilde{v}_{k+1}^{(m)}(\cdot, Q)$ is convex for any $Q \in \mathcal{Q}_c(t_{k+1})$, then using claim \ref{affine_cond_drift} (in Proposition \ref{cond_drit_dim1_refin}) and a straightforward induction one shows that
$$\mathbb{R} \ni x \mapsto \tilde{\mathrm{P}}^{(mn)}_{km:(k+1)m}\big(\tilde{v}_{k+1}^{(m)}(\cdot, Q+q) \big)(x) \quad \text{is convex}.$$
The latter, combined with the convexity Assumption $\bm{(\mathcal{H}_1^{c})}$ of the payoff function, yields:
\begin{align*}
\tilde{v}_k^{(m)}\big(\overline{x_ky_k}^{\lambda} , Q\big) &\le \lambda \underset{q \in \mathbb{A}_c(t_k, Q)}{\sup} \hspace{0.1cm} \Big[\Psi_k\big(t_{k}, q, x_{k}\big) +  \tilde{\mathrm{P}}^{(mn)}_{km:(k+1)m}\big(\tilde{v}_{k+1}^{(m)}(\cdot, Q+q) \big)(x_k) \Big]\\
& \quad + (1 - \lambda) \underset{q \in \mathbb{A}_c(t_k, Q)}{\sup} \hspace{0.1cm} \Big[\Psi_k\big(t_{k}, q, y_{k}\big) +  \tilde{\mathrm{P}}^{(mn)}_{km:(k+1)m}\big(\tilde{v}_{k+1}^{(m)}(\cdot, Q+q) \big)(y_k) \Big]\\
&= \lambda \tilde{v}_k^{(m)}\big(x_{k}, Q\big)  + (1 - \lambda) \tilde{v}_k^{(m)}\big(y_{k}, Q\big).
\end{align*}
This completes the proof.
\end{proof}

We proved that, given the semi-convexity assumption on $\sigma$, the propagation of convexity through the \emph{BDPP} holds when the involved random noise is truncated. The next step is to take the limit $m \to +\infty$ to propagate that convexity. This is stated in Proposition \ref{cnv_dim1} and relies on two ingredients: (1) the convergence of the truncated Euler scheme toward the actual one. (2) the Lipschitz continuous property of the swing value function.

We still set the threshold $s_h$ as in Proposition \ref{approx_dim1_conv}. For $i \in \{0,\ldots, mn\}$, we consider the processes $\Big(\overline{X}^{x, i}_{t_{\ell}^{(mn)}} \Big)_{i \le \ell \le mn}$ and $\Big(\widetilde{X}^{x, i}_{t_{\ell}^{(mn)}}\Big)_{i \le \ell \le mn}$ denoting the Euler scheme (still with step $h = \frac{T}{mn}$) and its truncated version starting at $x \in \mathbb{R}$ at time $t_i^{(mn)}$ respectively. Then, we have the following convergence result.

\begin{Proposition}
\label{result_trunc_euler}
For all $u \ge 1$ and $i \le \ell \in \{0,\ldots, mn\}$ we have, for any compact set $K \subset \mathbb{R}$,
\begin{equation}
\underset{x \in K}{\sup} \hspace{0.1cm} \Big\|\widetilde{X}^{x, i}_{t_{\ell}^{(mn)}} - \overline{X}^{x, i}_{t_{\ell}^{(mn)}}\Big\|_u \xrightarrow[m \to +\infty]{} 0.
\end{equation}
\end{Proposition}

\begin{proof}
Let $x \in K$ with $K \subset \mathbb{R}$ being a compact set. Using the proof of Proposition 4.1 in \cite{jourdain2023convex}, one may show that there exists a sequence $(c_m)_{m \ge 1}$ (not depending on $x$) such that:
\begin{equation}
    \label{borne_error_euler}
    \Big\|\widetilde{X}^{x, i}_{t_{\ell}^{(mn)}} - \overline{X}^{x, i}_{t_{\ell}^{(mn)}}\Big\|_u \le c_m(1 + |x|) \quad \text{with} \quad c_m \xrightarrow[m \to +\infty]{} 0.
\end{equation}
The result also holds true when taking the supremum for $x$ lying in the compact set $K$. This completes the proof.
\end{proof}

Then, we prove the propagation of convexity under the semi-convexity assumption of the volatility function as a result by propagation by taking the limit.

\begin{Proposition}
\label{cnv_dim1}
Let $c \in \{firm, pen\}$. Under Assumptions $\bm{(\mathcal{H}_1^{c})}$, $\bm{(\mathcal{H}_3)}$, $\bm{(\mathcal{H}_4^c)}$ and if, in addition, assumptions of Proposition \ref{approx_dim1_conv} hold true, then for all $k \in \{0,\ldots, n\}$, $Q \in \mathcal{Q}_c(t_k)$ and any compact set $K \subset \mathbb{R}$
$$ K \ni x \mapsto v_k(x, Q) \quad \text{is convex}.$$
\end{Proposition}

\begin{proof}
Since limits propagate convexity, it suffices to combine, by the triangle inequality, results in Proposition \ref{cvg_val_func_m}, and the following result:
$$\forall K \subset \mathbb{R} \hspace{0.1cm} \text{compact}, \hspace{0.2cm} Q \in \mathcal{Q}_c(t_k), \quad  \lim_{m\to\infty} \underset{x \in K}{\sup} \hspace{0.1cm} \big|\tilde{v}_k^{(m)}(x, Q) -  v_k^{(m)}(x, Q) \big| = 0$$ 
which we are going to prove using a backward induction on $k$. The result clearly holds true for $k = n$. Assume now it holds for $k+1$. For any $x \in K$ and $Q \in \mathcal{Q}_c(t_k)$, using successively the classic inequality, $\big|\underset{i \in I}{\sup} \hspace{0.1cm} a_i - \underset{i \in I}{\sup} \hspace{0.1cm} b_i\big| \le \underset{i \in I}{\sup} \hspace{0.1cm} |a_i-b_i|$, and the triangle inequality, one has:
\begin{align*}
\big|\tilde{v}_k^{(m)}(x, Q) -  v_k^{(m)}(x, Q) \big| &\le \underset{q \in \mathbb{A}_c(t_k, Q)}{\sup} \hspace{0.1cm} \Big| \tilde{\mathrm{P}}_{km:(k+1)m}^{(mn)}\big(\tilde{v}_{k+1}^{(m)}(\cdot, Q + q)\big)(x) - \mathrm{P}_{km:(k+1)m}^{(mn)}\big(v_{k+1}^{(m)}(\cdot, Q + q)\big)(x) \Big|\\
 &\le \underset{q \in \mathbb{A}_c(t_k, Q)}{\sup} \hspace{0.1cm} \Big| \tilde{\mathrm{P}}_{km:(k+1)m}^{(mn)}\big(\tilde{v}_{k+1}^{(m)}(\cdot, Q + q)\big)(x) - \tilde{\mathrm{P}}_{km:(k+1)m}^{(mn)}\big(v_{k+1}^{(m)}(\cdot, Q + q)\big)(x) \Big|\\
 & \quad + \underset{q \in \mathbb{A}_c(t_k, Q)}{\sup} \hspace{0.1cm} \Big| \tilde{\mathrm{P}}_{km:(k+1)m}^{(mn)}\big(v_{k+1}^{(m)}(\cdot, Q + q)\big)(x) - \mathrm{P}_{km:(k+1)m}^{(mn)}\big(v_{k+1}^{(m)}(\cdot, Q + q)\big)(x) \Big|.
\end{align*}
Let us deal with the r.h.s. sum term by term. We omit the supremum for $q \in \mathbb{A}_c(t_k, Q)$ owing to assumption $\bm{\mathcal{H}}$. Then, for $R \ge 1$, it is straightforward that, for the first term, there exists a positive constant $\kappa_{\beta, \sigma, k}^{(1)}$ such that:

\begin{align*}
    \Big|\tilde{\mathrm{P}}_{km:(k+1)m}^{(mn)}\big(&\tilde{v}_{k+1}^{(m)}(\cdot, Q + q)\big)(x) - \tilde{\mathrm{P}}_{km:(k+1)m}^{(mn)}\big(v_{k+1}^{(m)}(\cdot, Q + q)\big)(x) \Big|\\
    &= \Big|\mathbb{E}\big(\tilde{v}_{k+1}^{(m)}(\widetilde{X}_{t_{k+1}}^{(mn)}, Q+q) \rvert \widetilde{X}_{t_{k}}^{(mn)} = x \big) -\mathbb{E}\big(v_{k+1}^{(m)}(\widetilde{X}_{t_{k+1}}^{(mn)}, Q+q) \rvert \widetilde{X}_{t_{k}}^{(mn)} = x \big)\Big|\\
    &\le \Big\|\tilde{v}_{k+1}^{(m)}(\cdot, Q+q) - v_{k+1}^{(m)}(\cdot, Q+q) \Big\|_{\mathcal{B}(0, R)}\\
    & \quad + \mathbb{E}\Bigg|\Big(\tilde{v}_{k+1}^{(m)}\big(\widetilde{X}_{t_{k+1}}^{x, km}, Q+q\big) -v_{k+1}^{(m)}\big(\widetilde{X}_{t_{k+1}}^{x, km}, Q+q\big)  \Big) \cdot \mathbf{1}_{\big\{\big|\widetilde{X}_{t_{k+1}}^{x, km}\big| \ge R \big\}}  \Bigg|\\
    &\le \Big\|\tilde{v}_{k+1}^{(m)}(\cdot, Q+q) - v_{k+1}^{(m)}(\cdot, Q+q) \Big\|_{\mathcal{B}(0, R)} + \kappa_{\beta, \sigma, k}^{(1)} \cdot \mathbb{E}\Bigg(\Big(1 + \big|\widetilde{X}_{t_{k+1}}^{x, km}\big| \Big) \cdot \mathbf{1}_{\big\{\big|\widetilde{X}_{t_{k+1}}^{x, km}\big| \ge R \big\}} \Bigg),
\end{align*}
where in the last inequality, we used the Lipschitz property of functions $\tilde{v}_{k+1}^{(m)}(\cdot, Q+q), v_{k+1}^{(m)}(\cdot, Q+q)$ with their Lipschitz coefficients hidden in the positive constant $\kappa_{\beta, \sigma, k}^{(1)}$. Besides, note that for $R \ge 1$, using the Lipschitz property of $\beta_t, \sigma_t$ uniformly in $t$ and standard arguments, there exists a positive constant $\kappa_{\beta, \sigma, k}^{(2)}$ such that:
$$\mathbb{E}\Bigg(\Big(1 + \big|\widetilde{X}_{t_{k+1}}^{x, km}\big| \Big) \cdot \mathbf{1}_{\big\{\big|\widetilde{X}_{t_{k+1}}^{x, km}\big| \ge R \big\}} \Bigg) \le  \frac{2}{R} \cdot \Big\|\widetilde{X}_{t_{k+1}}^{x, km} \Big\|_2^2 \le \frac{\kappa_{\beta, \sigma, k}^{(2)}}{R} \big(1 + |x|^2\big).$$

Putting all together, for any $R \ge 1$, one has:
\begin{align*}
    \underset{x \in K}{\sup} \hspace{0.1cm} \Big|&\tilde{\mathrm{P}}_{km:(k+1)m}^{(mn)}\big(\tilde{v}_{k+1}^{(m)}(\cdot, Q + q)\big)(x) - \tilde{\mathrm{P}}_{km:(k+1)m}^{(mn)}\big(v_{k+1}^{(m)}(\cdot, Q + q)\big)(x) \Big|\\
    &\le \Big\|\tilde{v}_{k+1}^{(m)}(\cdot, Q+q) - v_{k+1}^{(m)}(\cdot, Q+q) \Big\|_{\mathcal{B}(0, R)} + \frac{\kappa_{\beta, \sigma, k}^{(1)}\kappa_{\beta, \sigma, k}^{(2)}}{R} \Big(1 + \underset{x \in K}{\sup} \hspace{0.1cm} |x|^2\Big).
\end{align*}
Since $K$ is compact set, first letting $m \to + \infty$ and using the induction assumption and then letting $R \to +\infty$ yields:
$$\lim_{m \to +\infty} \hspace{0.1cm} \underset{x \in K, q \in \mathbb{A}_c(t_k, Q)}{\sup} \hspace{0.1cm} \Big|\tilde{\mathrm{P}}_{km:(k+1)m}^{(mn)}\big(\tilde{v}_{k+1}^{(m)}(\cdot, Q + q)\big)(x) - \tilde{\mathrm{P}}_{km:(k+1)m}^{(mn)}\big(v_{k+1}^{(m)}(\cdot, Q + q)\big)(x) \Big|  =0.$$

We now handle the second term. One has:
\begin{align*}
    \Big| \tilde{\mathrm{P}}_{km:(k+1)m}^{(mn)}\big(&v_{k+1}^{(m)}(\cdot, Q + q)\big)(x) - \mathrm{P}_{km:(k+1)m}^{(mn)}\big(v_{k+1}^{(m)}(\cdot, Q + q)\big)(x) \Big|\\
    &=  \Big|\mathbb{E}\big(v_{k+1}^{(m)}(\widetilde{X}_{t_{k+1}}^{(mn)}, Q+q) \rvert \widetilde{X}_{t_{k}}^{(mn)} = x \big) -\mathbb{E}\big(v_{k+1}^{(m)}(\overline{X}_{t_{k+1}}^{(mn)}, Q+q) \rvert \overline{X}_{t_{k}}^{(mn)} = x \big)\Big|\\
    &= \Big|\mathbb{E}v_{k+1}^{(m)}\Big(\widetilde{X}_{t_{(k+1)m}^{(mn)}}^{x, km} \Big) - \mathbb{E}v_{k+1}^{(m)}\Big(\overline{X}_{t_{(k+1)m}^{(mn)}}^{x, k m}  \Big) \Big|\\
    &\le [v_{k+1}^{(m)}]_{Lip} \cdot \Bigg\|\widetilde{X}_{t_{(k+1)m}^{(mn)}}^{x, k m}-  \overline{X}_{t_{(k+1)m}^{(mn)}}^{x, k m}\Bigg\|_1 \xrightarrow[m \to +\infty]{} 0,
\end{align*}
where the convergence is uniform w.r.t. $x$ lying in the compact set $K$, owing to Remark \ref{bounded_lip_coeff_vk}, Proposition \ref{result_trunc_euler} and Proposition \ref{lispc_func_val}. This completes the proof.
\end{proof}

\section{Numerical experiments}
\label{sec_appl}
This section illustrates the main results of this paper, namely the convexity result and the functional monotonicity result. To this end, we consider a 15-day swing contract with daily exercise rights and a constant strike price set at $20$. The volume constraints configuration is given by parameters: $\underline{q} = 0$, $\overline{q} = 6$, $\underline{Q} = 50$, and $\overline{Q} = 80$. The swing physical space, representing the attainable cumulative consumption at each exercise date, is illustrated in Figure \ref{vol_curve_bgm}.

\begin{figure}[!ht]
\centering
\begin{tikzpicture}
\begin{axis}[
    xlabel={Exercise dates},
    ylabel={Cumulative consumption},
    xmin=0, xmax=16,
    ymin=-5, ymax=85,
    ymajorgrids=true,
    grid style=dashed,
   	width=0.5\linewidth,
	height=0.2\paperheight,
	legend pos = north west,
]

\addplot[color=blue, mark=triangle]
    coordinates {
    (0,0.0)(1,0.0)(2,0.0)(3,0.0)(4,0.0)(5,0.0)(6,0.0)(7,2.0)(8,8.0)(9,14.0)(10,20.0)(11,26.0)(12,32.0)(13,38.0)(14,44.0)(15,50.0)
    };
    \addlegendentry{Lower bound};
    
\addplot[color=orange, mark=triangle]
    coordinates {
    (0,0.0)(1,6.0)(2,12.0)(3,18.0)(4,24.0)(5,30.0)(6,36.0)(7,42.0)(8,48.0)(9,54.0)(10,60.0)(11,66.0)(12,72.0)(13,78.0)(14,80.0)(15,80.0)
    };
    \addlegendentry{Upper bound};
\end{axis}
\end{tikzpicture}
\caption{The swing volume grid (\emph{firm constraints}).}
\label{vol_curve_bgm}
\end{figure}
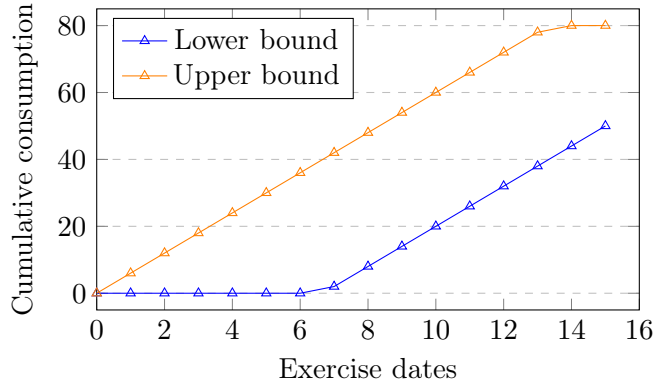

For the pricing, we use the Deep Neural Network (\emph{DNN}) approach introduced in \cite{lemaire2023swing}, and called \emph{NN strat}, which has shown competitive performance compared to state-of-the-art methods. Our neural network architecture consists of two hidden layers, of 10 units each. The Rectified Linear Unit (\emph{ReLU}) is used as activation function and a batch normalization is applied. We implement the \emph{DNN} using the \emph{PyTorch} \cite{Ketkar2017} toolbox and optimize it using the Preconditioned Stochastic Gradient Langevin Descent (PSGLD) method introduced in \cite{Li2015PreconditionedSG} which has already been used in the swing pricing context in \cite{lemaire2023swing} and generally analyzed and tested in \cite{PierreLangevin, bras:hal-03891234}. To perform the valuation, we use a Monte Carlo simulation with a sample size of $10^7$. Here, only the instantaneous price $S_t := F_{t,t}$ is considered. Generally, we will denote by $F_{t, T}$, the price of the underlying delivery contract seen at time $t$ with a commodity delivery at time $T \ge t$.

In the next two sections, we present numerical illustrations that are built upon a log-normal forward diffusion model. More precisely, we consider a $q$-factor forward diffusion model whose dynamics is given by:
\begin{equation}
\label{multi_fac_log_model}
	\frac{dF_{t, T}}{F_{t_, T}} = \sum_{i = 1}^{q}  \tilde{\sigma}_i e^{-\alpha_i (T-t)}dW_t^i, \quad t \in [0, T],
\end{equation}
where all $\big(W_t^i\big)_{t\ge 0}$, $1 \le i \le q$, are correlated Brownian motion i.e. with a square bracket:
\begin{equation*}
\langle dW_{\cdot}^i,dW_{\cdot}^j \rangle_t =  \left\{
    \begin{array}{ll}
        dt \hspace{1.4cm} \text{if} \hspace{0.2cm} i = j,\\
        \rho_{i, j} \cdot dt \hspace{0.6cm} \text{if} \hspace{0.2cm} i \neq j.
    \end{array}
\right.
\end{equation*}
Then, it follows from Itô's formula that
$$F_{t, T} = F_{0, T} \cdot \exp \Bigg(\sum_{i=1}^q  \tilde{\sigma}_i  \int_{0}^{t} e^{-\alpha_i (T-s)}\, \mathrm{d}W_s^i - \frac{1}{2} \sum_{i,j=1}^q  \tilde{\sigma}_i \tilde{\sigma}_j \int_{0}^{t} e^{-(\alpha_i + \alpha_j) (T-s)}\, \mathrm{d}s\Bigg)$$
so that, the spot price is given by:
\begin{equation}
\label{spot_price_log_n_model}
S_t = F_{0,t} \cdot \exp\Big(\langle \tilde{\sigma}, X_t \rangle - \frac{1}{2}\lambda_t^2\Big),
\end{equation}
where $\langle \tilde{\sigma}, X_t \rangle$ denotes the Euclidean inner product between vectors $\tilde{\sigma} = \big(\tilde{\sigma}_1, \ldots,  \tilde{\sigma}_q)^\top$ and $X_t = \big( X_t^1, \ldots, X_t^q \big)^\top $. Besides, for all $1 \le i \le q$:
\begin{equation*}
X_t^i = \int_{0}^{t} e^{-\alpha_i (t-s)}\, \mathrm{d}W_s^i \hspace{0.3cm} \text{and} \hspace{0.2cm} \lambda_t^2 = \sum_{i = 1}^{q} \frac{\tilde{\sigma}_i^2}{2\alpha_i}\big(1-e^{-2\alpha_i t}\big) +  \sum_{1 \le i \neq j \le q} \rho_{i, j} \frac{\tilde{\sigma}_i \tilde{\sigma}_j}{\alpha_i + \alpha_j} \Big(1 - e^{-(\alpha_i + \alpha_j)t}  \Big).
\end{equation*}

\subsubsection*{\underline{About the sensitivity: Delta computation}}
In what follows, we will represent swing prices as well as their sensitivities with respect to the forward price (see Figures \ref{one_fact_illus}, \ref{one_factor_pen}, \ref{result_3factors_model}) which we will call \emph{delta}. This is not the main purpose of the paper but we use it to, first double-check the convexity propagation, but also to emphasize the fact that the computation of this delta makes no problem in our numerical approach.

The computation of sensitivities boils down to the computation of the derivative of a certain function which, in our case, is the swing price defined by a stochastic optimal control problem. To achieve this, we rely on the \emph{envelope theorem} (see \cite{Milgrom2002EnvelopeTF}, and as used in \cite{lemaire2023swing}) which state that, under mild assumptions and in a parametric optimization problem, the sensitivity of the optimal value with respect to a parameter is given by the direct partial derivative of the objective function, evaluated at the optimal strategy, ignoring how the optimal choice itself changes. In our model \eqref{spot_price_log_n_model}, this derivative is given by (in case of a standard swing contract without penalty):
$$\nabla_{(F_{0,t_0},\ldots, F_{0,t_{n-1}})} \pi_0 = \Bigg(q_0^{*} \frac{S_{t_0}}{F_{0, t_0}}, \ldots,  q_{n-1}^{*}\frac{S_{t_{n-1}}}{F_{0, t_{n-1}}} \Bigg)^\top,$$
where $(q_0^{*},\ldots, q_{n-1}^{*})$ denotes the optimal control and $\pi_0$ is the swing initial price. For example, for a standard swing contract without penalty, with a fixed strike, and with no interest rate, this price is formally given by:
$$\pi_0 := \underset{(q_0, \ldots, q_{n-1}) \in \mathcal{A}}{\sup} \hspace{0.1cm} \mathbb{E}\Bigg[\sum_{k=0}^{n-1} q_k \big(S_{t_k} - K\big) \Bigg] = \underset{(q_0, \ldots, q_{n-1}) \in \mathcal{A}}{\sup} \hspace{0.1cm} \mathbb{E}\left[\sum_{k=0}^{n-1} q_k \Bigg(F_{0,t_k} \cdot \exp\Big(\langle \tilde{\sigma}, X_{t_k} \rangle - \frac{1}{2}\lambda_{t_k}^2\Big) - K\Bigg) \right],$$
where $\mathcal{A}$ is the set of admissible controls. Besides, in the sequel, only for the sake of presentation, we assume:
$$\forall t \ge 0, \quad F_{0, t} = F_{0} > 0.$$

\subsection{One factor model}
We start by a one-factor log-normal model, that is $q = 1$. The dynamics \eqref{multi_fac_log_model} reads,
\begin{equation}
	\frac{dF_{t, T}}{F_{t_, T}} = \tilde{\sigma} e^{-\alpha (T-t)}dW_t, \quad  t \le T,
	\label{hjm_model_one_factor}
\end{equation}
where $W$ is a standard Brownian motion and,
\begin{equation}
\label{spot_model}
S_t = F_{0,t} \cdot \exp\big(\tilde{\sigma} X_t - \frac{1}{2}\lambda_t^2\big) \hspace{0.2cm} \text{with} \hspace{0.2cm} X_t = \int_{0}^{t} e^{-\alpha (t-s)}\, \mathrm{d}W_s \hspace{0.3cm} \text{and} \hspace{0.3cm} \lambda_t^2 = \frac{\tilde{\sigma}^2}{2\alpha}\big(1-e^{-2\alpha t}\big).
\end{equation}
The Euler-Maruyama scheme of the diffusion \eqref{hjm_model_one_factor} writes,
\begin{equation*}
F_{t_{k+1}, T} = F_{t_{k}, T} + \sigma_{\tilde{\sigma}, k}\big(F_{t_{k}, T} \big) Z_{k+1} \hspace{0.4cm} \text{with} \hspace{0.4cm} \sigma_{\tilde{\sigma}, k}(x) = \tilde{\sigma} x\sqrt{\Delta t_k} e^{-\alpha (T-t_k)},
\end{equation*}
where $\Delta t_k = t_{k+1} - t_k$ and $(Z_k)_k$ are \emph{i.i.d.} copies of $Z \stackrel{\mathcal{L}}{\sim} \mathcal{N}(0,1)$. This model falls into a (martingale) \emph{ARCH} dynamics. Besides $\sigma_{\tilde{\sigma}, k}$ is affine (hence convex) and for $0 <\tilde{\sigma}_1 \le \tilde{\sigma}_2$, we clearly have $|\sigma_{\tilde{\sigma}_1, k}(x)|  \le |\sigma_{\tilde{\sigma}_2, k}(x)|$ meaning that $\sigma_{\tilde{\sigma}_1, k}(x) \preceq \sigma_{\tilde{\sigma}_2, k}(x)$. Hence assumptions for the functional monotonicity and convexity propagation hold true for this model. Our numerical results are illustrated in Figure \ref{one_fact_illus} with $\alpha = 0.4, \tilde{\sigma}_1 = 0.2, \tilde{\sigma}_2 = 0.7$. We see that the swing price is increasing with the volatility parameter $\Tilde{\sigma}$ and the first partial derivative of the swing price, confirming the convexity.
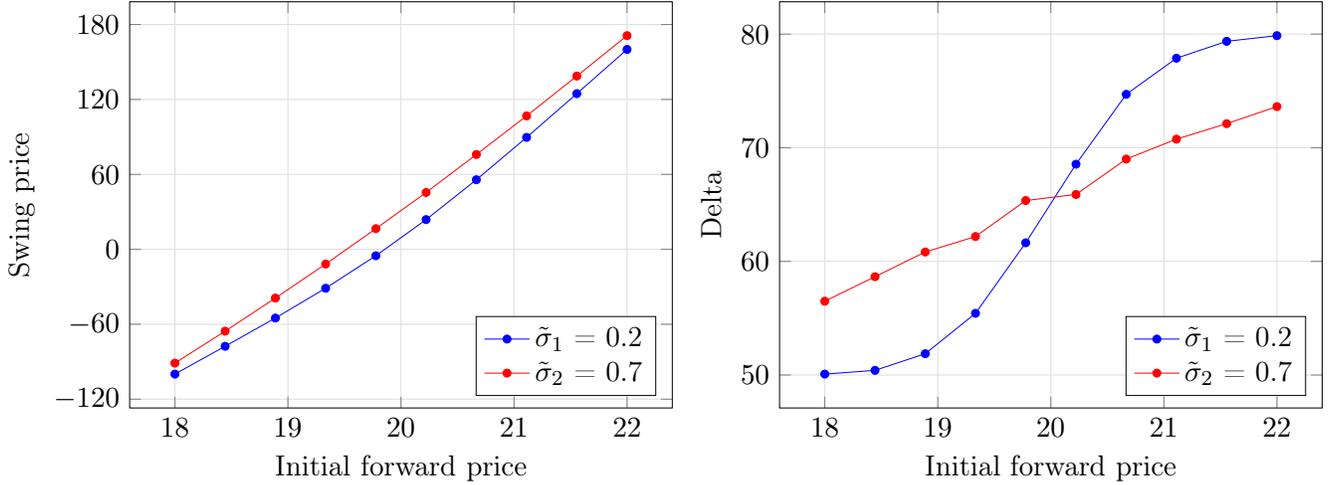
\begin{figure}[ht]
\centering

\begin{tikzpicture}
\begin{axis}[
	xlabel=Initial forward price $F_0$,
	ylabel=Swing price,
	grid=both,
        ytick distance = 60,
	minor grid style={gray!25},
	major grid style={gray!25},
	width=0.5\linewidth,
	height=0.25\paperheight,
	line width=0.2pt,
	mark size=1.5pt,
	mark options={solid},
	legend pos =  south east
	]
\addplot[color=blue, mark=*] %
	table[x=f0,y=price,col sep=comma]{datas/one_factor_price_delta_vol_0.2.csv};
\addlegendentry{$\tilde{\sigma}_1$ = 0.2};
\addplot[color=red, mark=*] %
	table[x=f0,y=price,col sep=comma]{datas/one_factor_price_delta_vol_0.7.csv};
\addlegendentry{$\tilde{\sigma}_2$ = 0.7};
\end{axis}
\end{tikzpicture}%
~
\begin{tikzpicture}
\begin{axis}[
	xlabel=Initial forward price $F_0$,
	ylabel=Delta,
	grid=both,
	minor grid style={gray!25},
	major grid style={gray!25},
	width=0.5\linewidth,
	height=0.25\paperheight,
	line width=0.2pt,
	mark size=1.5pt,
	mark options={solid},
	legend pos =  south east
	]
\addplot[color=blue, mark=*] %
	table[x=f0,y=delta,col sep=comma]{datas/one_factor_price_delta_vol_0.2.csv};
\addlegendentry{$\tilde{\sigma}_1$ = 0.2};
\addplot[color=red, mark=*] %
	table[x=f0,y=delta,col sep=comma]{datas/one_factor_price_delta_vol_0.7.csv};
\addlegendentry{$\tilde{\sigma}_2$ = 0.7};
\end{axis}
\end{tikzpicture}
\caption{Swing price and its delta in terms of initial forward price for different values of $\tilde{\sigma}$.}
\label{one_fact_illus}
\end{figure}
We then consider the \emph{call} version of the standard swing payoff (see \eqref{call_payoff}).
The swing prices are depicted in \ref{one_fact_illus_call} with $\tilde{\sigma}_1 = 0.2, \tilde{\sigma}_2 = 0.4$.

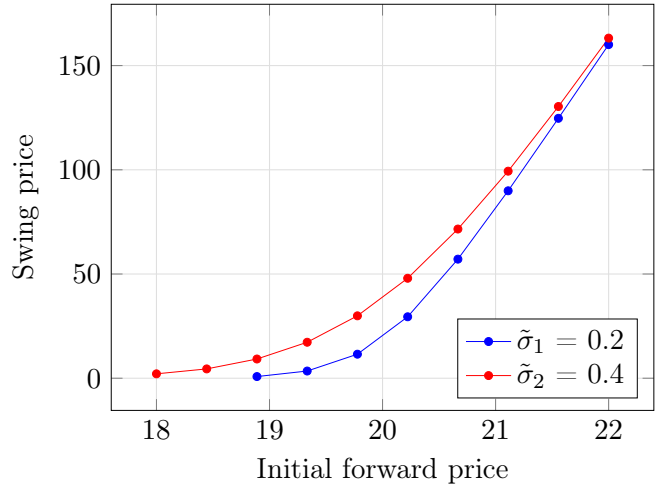
\begin{figure}[ht!]
\centering

\begin{tikzpicture}
\begin{axis}[
	xlabel=Initial forward price $F_0$,
	ylabel=Swing price,
	grid=both,
	minor grid style={gray!25},
	major grid style={gray!25},
	width=0.5\linewidth,
	height=0.25\paperheight,
	line width=0.2pt,
	mark size=1.5pt,
	mark options={solid},
	legend pos =  south east
	]
\addplot[color=blue, mark=*] %
	table[x=f0,y=price,col sep=comma]{datas/one_factor_price_delta_vol_0.2_call.csv};
\addlegendentry{$\tilde{\sigma}_1$ = 0.2};
\addplot[color=red, mark=*] %
	table[x=f0,y=price,col sep=comma]{datas/one_factor_price_delta_vol_0.4_call.csv};
\addlegendentry{$\tilde{\sigma}_2$ = 0.4};
\end{axis}
\end{tikzpicture}%
\caption{Swing price and its delta in terms of initial forward price for the call payoff \eqref{call_payoff}.}
\label{one_fact_illus_call}
\end{figure}
The results presented in Figures \ref{one_fact_illus} and \ref{one_fact_illus_call} provide numerical evidence that the swing price is convex with respect to the initial forward price $F_0$. Besides, it is also noteworthy that the convex shape of the swing price becomes more pronounced when the payoff exhibits a stronger convex shape as when we used the \emph{call} payoff \eqref{call_payoff} (see Figure \ref{one_fact_illus_call}).

The same numerical illustrations are performed for the penalty setting using $A = B = 0.2$ for the penalty function (see \eqref{penalty_func}). Results are depicted in Figure \ref{one_factor_pen} with $\tilde{\sigma}_1 = 0.3, \tilde{\sigma}_2 = 0.7$.
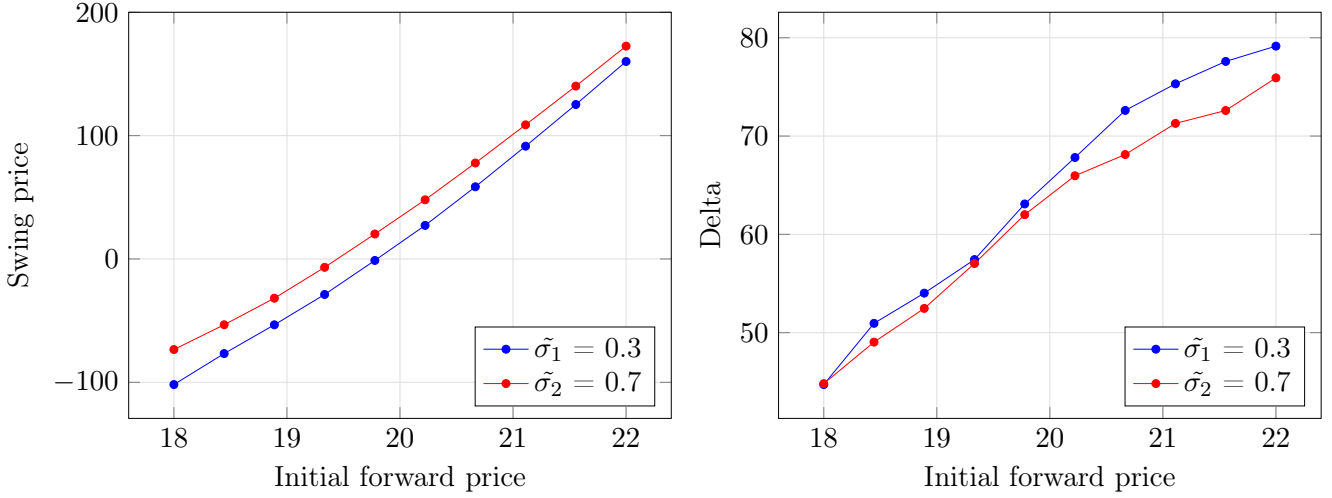
\begin{figure}[!h]
\centering
\begin{tikzpicture}
\begin{axis}[
	xlabel=Initial forward price $F_0$,
	ylabel=Swing price,
	grid=both,
	minor grid style={gray!25},
	major grid style={gray!25},
	width=0.5\linewidth,
	height=0.25\paperheight,
	line width=0.2pt,
	mark size=1.5pt,
	mark options={solid},
	legend pos =  south east
	]
\addplot[color=blue, mark=*] %
	table[x=f0,y=price,col sep=comma]{datas/penalty/one_factor_price_delta_vol_0.3_pen.csv};
\addlegendentry{$\tilde{\sigma_1}$ = 0.3};
\addplot[color=red, mark=*] %
	table[x=f0,y=price,col sep=comma]{datas/penalty/one_factor_price_delta_vol_0.7_pen.csv};
\addlegendentry{$\tilde{\sigma_2}$ = 0.7};
\end{axis}
\end{tikzpicture}%
~
\begin{tikzpicture}
\begin{axis}[
	xlabel=Initial forward price $F_0$,
	ylabel=Delta,
	grid=both,
	minor grid style={gray!25},
	major grid style={gray!25},
	width=0.5\linewidth,
	height=0.25\paperheight,
	line width=0.2pt,
	mark size=1.5pt,
	mark options={solid},
	legend pos =  south east
	]
\addplot[color=blue, mark=*] %
	table[x=f0,y=delta,col sep=comma]{datas/penalty/one_factor_price_delta_vol_0.3_pen.csv};
\addlegendentry{$\tilde{\sigma_1}$ = 0.3};
\addplot[color=red, mark=*] %
	table[x=f0,y=delta,col sep=comma]{datas/penalty/one_factor_price_delta_vol_0.7_pen.csv};
\addlegendentry{$\tilde{\sigma_2}$ = 0.7};
\end{axis}
\end{tikzpicture}
\caption{Prices and forward delta in terms of initial forward price (swing with penalty). We considered the one factor model \eqref{hjm_model_one_factor}.}
\label{one_factor_pen}
\end{figure}

\subsection{Multi-factor model}
We finally consider a three factor model (i.e., $q=3$) whose dynamics is given by,
\begin{equation}
	\label{3fac_diff}
	\frac{dF_{t, T}}{F_{t_, T}} = \sum_{i = 1}^{3}  \tilde{\sigma}_i e^{-\alpha_i (T-t)}dW_t^i.
\end{equation}
Here we set $\alpha_i = \alpha = 0.8, \tilde{\sigma}_i = 0.7, \rho_{i, j} = \rho \in [-1, 1]$. The Euler-Maruyama scheme of \eqref{3fac_diff} is,
\begin{equation*}
F_{t_{k+1}, T} = F_{t_{k}, T} +  \sigma_{\rho, k}\big(F_{t_k, T}\big) \cdot Z_{k+1},
\end{equation*}
where the matrix-valued function $\sigma_{\rho}\big(t_k, \cdot\big)$ is given by,
\begin{equation}
\label{mat_sig_3fac}
\sigma_{\rho, k}(x) = \Big(x \sqrt{\Delta t_k} \tilde{\sigma}_1 e^{-\alpha_1 (T- t_k)}, \ldots, x \sqrt{\Delta t_k} \tilde{\sigma}_q e^{-\alpha_q (T- t_k)} \Big) \cdot L(\rho) \in \mathbb{M}_{1, q}\big(\mathbb{R}\big)
\end{equation}
with $L(\rho) = \big(L_{i, j}(\rho)\big)_{1 \le i, j \le q}$ being the Cholesky decomposition of the correlation matrix $\Gamma := \big(\rho_{i, j}\big)_{1 \le i, j \le q}$ given by (assuming $\rho > -\frac{1}{q-1}$),
\begin{equation}
\label{mat_gamma}
\Gamma(\rho) = \big[\rho + (1-\rho)\mathbf{1}_{i = j}  \big]_{1 \le i,j \le q} = \begin{pmatrix} 
         1&\rho& \cdots & \cdots & \rho \\ 
         \rho &1& \ddots &  & \rho \\
         \vdots & \ddots & \ddots & \ddots & \vdots\\ 
         \vdots & & \ddots & 1 & \rho \\ 
         \rho & \cdots & \cdots & \rho & 1
\end{pmatrix}\in \mathcal{S}^{+}\big(q, \mathbb{R} \big)
\end{equation}
and $(Z_k)_k$ being \emph{i.i.d.} copies of $Z \sim \mathcal{N}(0,\mathbf{I}_3)$. For this specific matrix $\Gamma(\rho)$ \eqref{mat_gamma}, its Cholesky decomposition $L(\rho)$ has an explicit form given by Proposition \ref{chol_dec_prop}. Besides owing to Proposition \ref{mono_corr}, the matrix-valued function \eqref{mat_sig_3fac} meets the domination criterion when parameter $\rho$ varies. That is, if $\rho_1 \le \rho_2$ then $\sigma_{\rho_1, k}(x)  \preceq \sigma_{\rho_2, k}(x)$. It remains to discuss the $\preceq$-convexity of $\sigma_{\rho, k}$. It suffices to note that $\sigma_{\rho, k}$ can be written as in Remark \ref{gen_ex_cnv_matrix}. Indeed, one has
$$\sigma_{\rho, k}(x) = A \cdot \underbrace{\diag\Big(x\cdot \sqrt{\Delta t_k}, \ldots, x \cdot \sqrt{\Delta t_k} \Big)}_{\in \mathbb{M}_{q,q}\big(\mathbb{R}\big)}$$
with $A := \big(\tilde{\sigma}_1 e^{-\alpha_1 (T- t_k)}, \ldots, \tilde{\sigma}_q e^{-\alpha_q (T- t_k)} \big) \cdot L(\rho) \in \mathbb{M}_{1, q}\big(\mathbb{R}\big)$. Therefore, with the notations of Remark \ref{gen_ex_cnv_matrix}, it suffices to set $O = I_q \in \mathcal{O}(q, \mathbb{R})$. This shows that $\sigma_{\rho, k}$ is $\preceq$-convex. Thus, by the domination criterion, the swing price in model \eqref{3fac_diff} is increasing with the correlation parameter $\rho$. This claims is in line with what can be observed in the numerical illustrations in \cite{lemaire2023swing}. Results are shown in Figures \ref{result_3factors_model} and \ref{result_3factors_model_call} with $\rho_1 = 0.1, \rho_2 = 0.4$.

\begin{figure}[ht]
\centering

\begin{tikzpicture}
\begin{axis}[
	xlabel=Initial forward price $F_0$,
	ylabel=Swing price,
	grid=both,
	minor grid style={gray!25},
	major grid style={gray!25},
	width=0.5\linewidth,
	height=0.25\paperheight,
	line width=0.2pt,
	mark size=1.5pt,
	mark options={solid},
	legend pos =  south east
	]
\addplot[color=blue, mark=*] %
	table[x=f0,y=price,col sep=comma]{datas/3_factor_price_delta_corr0.1.csv};
\addlegendentry{$\rho_1$ = 0.1};
\addplot[color=red, mark=*] %
	table[x=f0,y=price,col sep=comma]{datas/3_factor_price_delta_corr0.4.csv};
\addlegendentry{$\rho_2$ = 0.4};
\end{axis}
\end{tikzpicture}%
~
\begin{tikzpicture}
\begin{axis}[
	xlabel=Initial forward price $F_0$,
	ylabel=Delta,
	grid=both,
	minor grid style={gray!25},
	major grid style={gray!25},
	width=0.5\linewidth,
	height=0.25\paperheight,
	line width=0.2pt,
	mark size=1.5pt,
	mark options={solid},
	legend pos =  south east
	]
\addplot[color=blue, mark=*] %
	table[x=f0,y=delta,col sep=comma]{datas/3_factor_price_delta_corr0.1.csv};
\addlegendentry{$\rho_1$ = 0.1};
\addplot[color=red, mark=*] %
	table[x=f0,y=delta,col sep=comma]{datas/3_factor_price_delta_corr0.4.csv};
\addlegendentry{$\rho_2$ = 0.4};
\end{axis}
\end{tikzpicture}
\caption{Swing price and its delta in terms of initial forward price for different correlation parameters $\rho$.}
\label{result_3factors_model}
\end{figure}
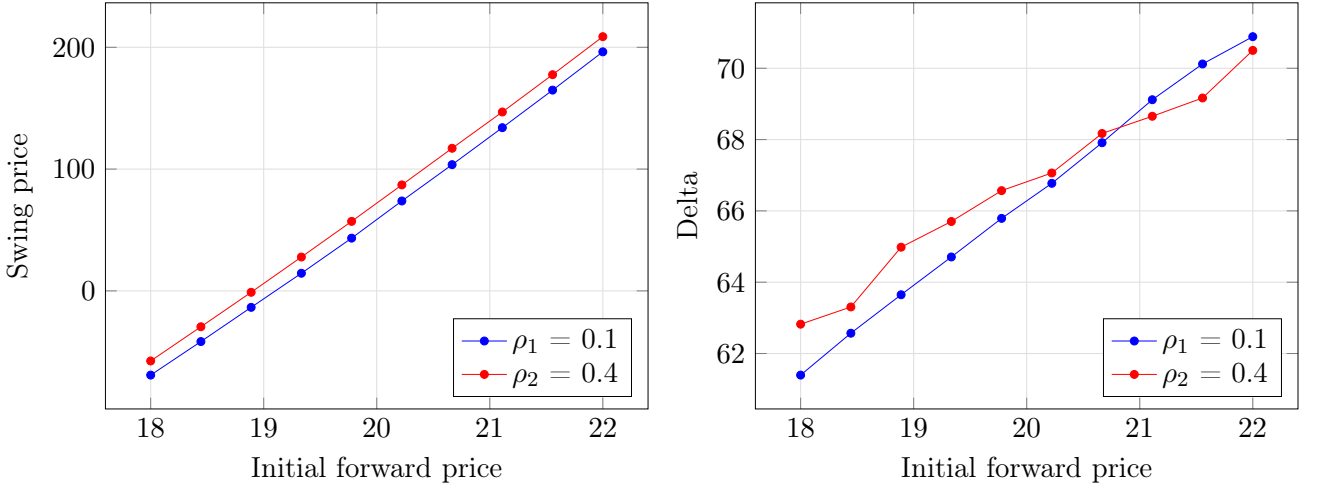

\begin{figure}[ht!]
\centering

\begin{tikzpicture}
\begin{axis}[
	xlabel=Initial forward price $F_0$,
	ylabel=Swing price,
	grid=both,
	minor grid style={gray!25},
	major grid style={gray!25},
	width=0.5\linewidth,
	height=0.25\paperheight,
	line width=0.2pt,
	mark size=1.5pt,
	mark options={solid},
	legend pos =  south east
	]
\addplot[color=blue, mark=*] %
	table[x=f0,y=price,col sep=comma]{datas/3_factor_price_delta_corr0.1_call.csv};
\addlegendentry{$\rho_1$ = 0.1};
\addplot[color=red, mark=*] %
	table[x=f0,y=price,col sep=comma]{datas/3_factor_price_delta_corr0.4_call.csv};
\addlegendentry{$\rho_2$ = 0.4};
\end{axis}
\end{tikzpicture}%
\caption{Swing price and its delta in terms of initial forward price for the call payoff \eqref{call_payoff}.}
\label{result_3factors_model_call}
\end{figure}
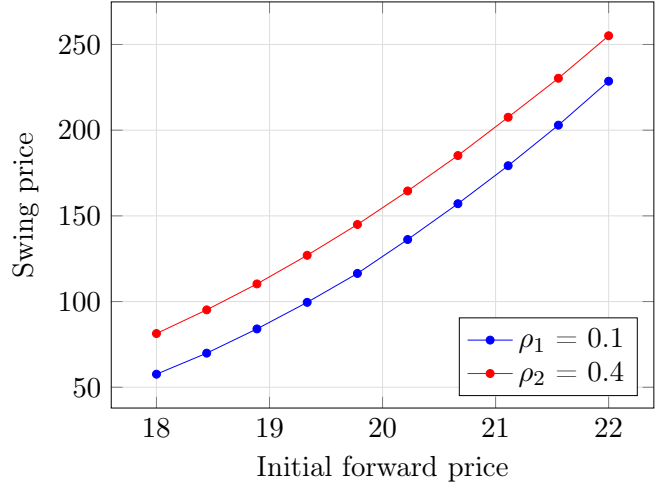
We have studied domination criterion with respect to the correlation parameter $\rho$. However, one may also prove a domination result with respect to the volatility parameters $(\tilde{\sigma}_j)_j$. Since the $\preceq$-convexity of $\sigma_{\rho, k}$ has already been shown, it remains to prove the $\preceq$-monotonicity  of the matrix-valued function $\sigma_{\rho, k}$ with respect to parameters $(\tilde{\sigma}_j)_j$ (using the pointwise order). This holds owing to Proposition \ref{vol_mono}. Indeed, keeping in mind equation \eqref{mat_sig_3fac} and setting,
$$A\big(\tilde{\sigma}_1, \ldots, \tilde{\sigma}_q \big) :=  \Big(x \sqrt{\Delta t_k} \tilde{\sigma}_1 e^{-\alpha_1 (T- t_k)}, \ldots, x \sqrt{\Delta t_k} \tilde{\sigma}_q e^{-\alpha_q (T- t_k)} \Big) \in \mathbb{M}_{1, q}\big(\mathbb{R}\big),$$
one may deduce, owing to Proposition \ref{vol_mono}, that if $\tilde{\sigma}_j \le \tilde{\sigma}_j'$ for all $1 \le j \le q$ then $A\big(\tilde{\sigma}_1, \ldots, \tilde{\sigma}_q \big) \preceq A\big(\tilde{\sigma}_1', \ldots, \tilde{\sigma}_q' \big)$. Results are illustrated in Figure \ref{result_3factors_model_sensi_vol}. We used $\rho = 0.3$.

\begin{figure}[ht]
\centering

\begin{tikzpicture}
\begin{axis}[
	xlabel=Initial forward price $F_0$,
	ylabel=Swing price,
	grid=both,
	minor grid style={gray!25},
	major grid style={gray!25},
	width=0.5\linewidth,
	height=0.25\paperheight,
	line width=0.2pt,
	mark size=1.5pt,
	mark options={solid},
	legend pos = north west,
	legend style={nodes={scale=0.7, transform shape}}
	]
\addplot[color=blue, mark=*] %
	table[x=f0,y=price,col sep=comma]{datas/3_factor_price_delta_sensi_vol_1.csv};
\addlegendentry{$\sigma_1 = 0.1$, $\sigma_2= 0.2, \sigma_3 = 0.3$};
\addplot[color=red, mark=*] %
	table[x=f0,y=price,col sep=comma]{datas/3_factor_price_delta_sensi_vol_2.csv};
\addlegendentry{$\sigma_1 = 0.4$, $\sigma_2= 0.4, \sigma_3 = 0.5$};
\end{axis}
\end{tikzpicture}%
~
\begin{tikzpicture}
\begin{axis}[
	xlabel=Initial forward price $F_0$,
	ylabel=Swing price (call payoff),
	grid=both,
	minor grid style={gray!25},
	major grid style={gray!25},
	label style={font=\small},
	width=0.5\linewidth,
	height=0.25\paperheight,
	line width=0.2pt,
	mark size=1.5pt,
	mark options={solid},
	legend pos =  north west,
	legend style={nodes={scale=0.7, transform shape}}
	]
\addplot[color=blue, mark=*] %
	table[x=f0,y=price,col sep=comma]{datas/3_factor_price_delta_sensi_vol_1_call.csv};
\addlegendentry{$\sigma_1 = 0.1$, $\sigma_2= 0.2, \sigma_3 = 0.3$};
\addplot[color=red, mark=*] %
	table[x=f0,y=price,col sep=comma]{datas/3_factor_price_delta_sensi_vol_2_call.csv};
\addlegendentry{$\sigma_1 = 0.4$, $\sigma_2= 0.4, \sigma_3 = 0.5$};
\end{axis}
\end{tikzpicture}
\caption{Swing price (left) and swing price with call payoff (right) in terms of initial forward price.}
\label{result_3factors_model_sensi_vol}
\end{figure}
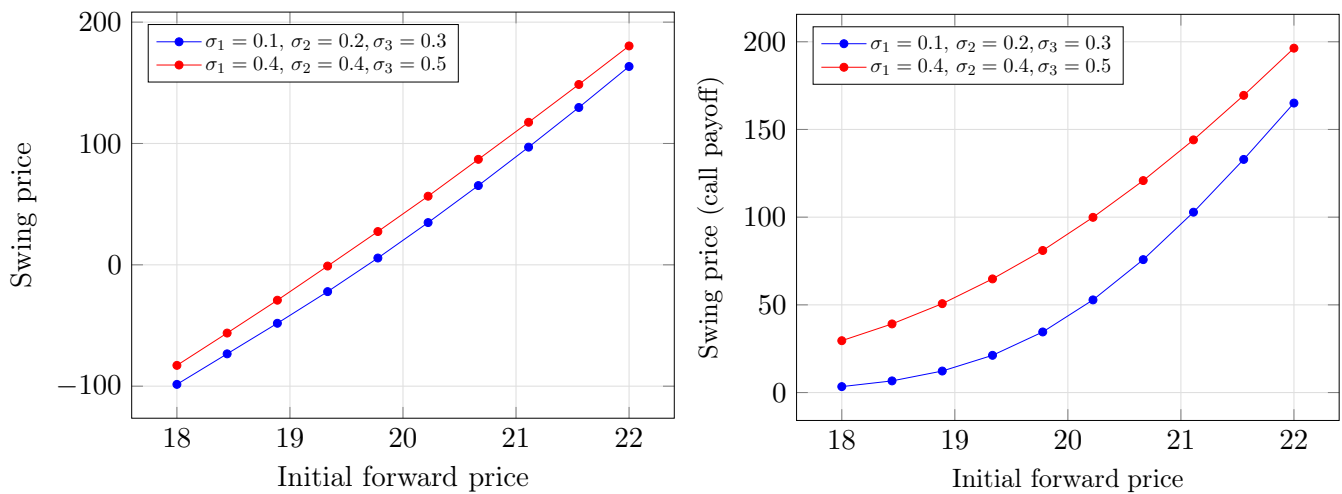

\subsection*{Acknowledgments}
The authors would like to thank Benjamin Jourdain, Vincent Lemaire, and Asma Meziou for fruitful discussions. The PhD thesis of the second author is funded by Engie Global Markets. The first author benefited from the support of the \q{Chaire Risques Financiers}, Fondation du Risque.

\bibliographystyle{plain}
\bibliography{biblio.bib}

\begin{thebibliography}{10}

\bibitem{10.1214/19-AIHP1014}
Aur{\'e}lien Alfonsi, Jacopo Corbetta, and Benjamin Jourdain.
\newblock Sampling of probability measures in the convex order by wasserstein projection.
\newblock {\em Annales de l'Institut Henri Poincaré, Probabilités et Statistiques}, 56(3):1706 -- 1729, 2020.

\bibitem{Bardou2009OptimalQF}
Olivier Bardou, Sandrine Bouthemy, and Gilles Pag{\`e}s.
\newblock Optimal quantization for the pricing of swing options.
\newblock {\em Applied Mathematical Finance}, 16:183--217, 2009.

\bibitem{Bardou2007WhenAS}
Olivier Bardou, Sandrine Bouthemy, and Gilles Pag{\`e}s.
\newblock When are swing options bang-bang?
\newblock {\em International Journal of Theoretical and Applied Finance (IJTAF)}, 13:867--899, 09 2010.

\bibitem{BarreraEsteve2006NumericalMF}
Christophe Barrera-Esteve, Florent Bergeret, Charles Dossal, Emmanuel Gobet, Asma Meziou, R{\'e}mi Munos, and Damien Reboul-Salze.
\newblock Numerical methods for the pricing of swing options: A stochastic control approach.
\newblock {\em Methodology and Computing in Applied Probability}, 8:517--540, 2006.

\bibitem{Bergenthu}
Jan Bergenthum and Ludger {Rüschendorf}.
\newblock Comparison results for path-dependent options.
\newblock {\em Statistics \& Decisions}, 26(1):53--72, 2008.

\bibitem{bras:hal-03891234}
Pierre Bras and Gilles Pagès.
\newblock Convergence of langevin-simulated annealing algorithms with multiplicative noise ii: Total variation.
\newblock {\em Monte Carlo Methods and Applications}, 29(3):203--219, 2023.

\bibitem{PierreLangevin}
Pierre Bras and Gilles Pagès.
\newblock Langevin algorithms for markovian neural networks and deep stochastic control.
\newblock In {\em 2023 International Joint Conference on Neural Networks (IJCNN)}, pages 1--8, 2023.

\bibitem{BURGERT2006289}
Christian Burgert and Ludger {Rüschendorf}.
\newblock Consistent risk measures for portfolio vectors.
\newblock {\em Insurance: Mathematics and Economics}, 38(2):289--297, 2006.

\bibitem{BAUERLE2006132}
Nicole {Bäuerle} and Alfred {Müller}.
\newblock Stochastic orders and risk measures: Consistency and bounds.
\newblock {\em Insurance: Mathematics and Economics}, 38(1):132--148, 2006.

\bibitem{Gupta_cvx_actuarial_science}
Arjun Gupta and Mohammad Aziz.
\newblock Convex ordering of random variables and its applications in econometrics and actuarial science.
\newblock {\em European Journal of Pure and Applied Mathematics [electronic only]}, 3, 01 2010.

\bibitem{jourdain2023convex}
Benjamin Jourdain and Gilles Pag{\`e}s.
\newblock Convex ordering of solutions to one-dimensional sdes.
\newblock {\em ArXiv}, 2312.09779, 2023.

\bibitem{jourdain:hal-02304190}
Benjamin Jourdain and Gilles Pagès.
\newblock {Convex Order, Quantization and Monotone Approximations of ARCH Models}.
\newblock {\em Journal of Theoretical Probability}, 35(4):2480--2517, December 2022.

\bibitem{jourdain:hal-03862241}
Benjamin Jourdain and Gilles Pagès.
\newblock Convex ordering for stochastic {V}olterra equations and their euler schemes.
\newblock {\em Finance and Stochastics}, 29:1–62, 2025.

\bibitem{Ketkar2017}
Nikhil Ketkar.
\newblock {\em Introduction to PyTorch}, pages 195--208.
\newblock Apress, Berkeley, CA, 2017.

\bibitem{lemaire2023swing}
Vincent Lemaire, Gilles Pagès, and Christian Yeo.
\newblock Swing contract pricing: With and without neural networks.
\newblock {\em Frontiers of Mathematical Finance}, 3(2):270--303, 2024.

\bibitem{Li2015PreconditionedSG}
Chunyuan Li, Changyou Chen, David~Edwin Carlson, and Lawrence Carin.
\newblock Preconditioned stochastic gradient {L}angevin dynamics for deep neural networks.
\newblock In {\em AAAI Conference on Artificial Intelligence}, 2015.

\bibitem{LIONS2006915}
Pierre-Louis Lions and Marek Musiela.
\newblock Convexity of solutions of parabolic equations.
\newblock {\em Comptes Rendus Mathematique}, 342(12):915--921, 2006.

\bibitem{liu2022functional}
Yating Liu and Gilles Pag{\`e}s.
\newblock Functional convex order for the scaled {M}c{K}ean--{V}lasov processes.
\newblock {\em Ann. Appl. Probab.}, 33(6A):4491--4527, 2023.

\bibitem{LIU2022312}
Yating Liu and Gilles Pagès.
\newblock M{onotone convex order for the McKean–Vlasov processes}.
\newblock {\em Stochastic Processes and their Applications}, 152:312--338, 2022.

\bibitem{Milgrom2002EnvelopeTF}
Paul~R. Milgrom and Ilya Segal.
\newblock Envelope theorems for arbitrary choice sets.
\newblock {\em Econometrica}, 70:583--601, 2002.

\bibitem{Pages_cvx_ord_path_dep}
Gilles Pag{\`e}s.
\newblock {\em Convex Order for Path-Dependent Derivatives: A Dynamic Programming Approach}, pages 33--96.
\newblock Springer International Publishing, Cham, 2016.

\bibitem{Pages2018}
Gilles Pag{\`e}s.
\newblock {\em Numerical Probability: An Introduction with Applications to Finance}, chapter~7, pages 271--361.
\newblock Springer International Publishing, Cham, 2018.

\bibitem{ROCKAFELLAR20021443}
Tyrrell Rockafellar and Stanislav Uryasev.
\newblock Conditional {V}alue-at-{R}isk for general loss distributions.
\newblock {\em Journal of Banking \& Finance}, 26(7):1443--1471, 2002.

\bibitem{article_Rockafellar_deviation_risk_meas}
Tyrrell Rockafellar, Stanislav Uryasev, and Zabarankin Michael.
\newblock Deviation measures in risk analysis and optimization.
\newblock {\em SSRN Electronic Journal}, 2002.

\end{thebibliography}

\appendix
\section{Some useful results}
\label{useful_res_appendix}
\begin{lemma}[Consistency \cite{10.1214/19-AIHP1014}]
\label{rq_consistency}
Expectations in \eqref{cvx_ord_def} are always well-defined (in $(- \infty, + \infty]$). Indeed, for all $x \in \mathbb{R}^d$, since $f$ is convex, we have
$$f(x) \ge f(0) + \langle \nabla_s f(0), x \rangle,$$
where $\nabla_s f(0)$ denotes a subgradient of $f$ at $0$. Then applying the function $x \in \mathbb{R} \mapsto x^{-} := \max(-x, 0)$ on both sides of the last inequality yields,
$$f^{-}(x) \le \big(f(0) + \langle \nabla_s f(0), x \rangle \big)^{-} \le \big| f(0) + \langle \nabla_s f(0), x \rangle \big|\le \big| f(0) \big| + \big|\langle \nabla_s f(0), x \rangle \big| \le \big| f(0) \big| + \big| \nabla_s f(0) \big| \cdot \big| x \big|,$$
where we successively used triangle inequality and then Cauchy-Schwartz inequality. Thus for every $U \in \mathbb{L}^1_{\mathbb{R}^d}\big( \mathbb{P} \big)$ we have,
$$ \mathbb{E}f^{-}(U)\le \big| f(0) \big| + \big| \nabla_s f(0) \big| \cdot \mathbb{E}| U | < + \infty.$$
Therefore,
$$\mathbb{E}f(U)  = \underbrace{\mathbb{E}f^{+}(U)}_{\in [0, + \infty]} - \underbrace{\mathbb{E}f^{-}(U)}_{\in [0, +\infty)} \in (-\infty, +\infty],$$
where $x^{+} := \max(x, 0)$.
\end{lemma}

\begin{lemma}[Stein lemma]
\label{stein_lem}
Suppose $Z \sim \mathcal{N}(\mu, \sigma^2)$. Then consider a $\mathcal{C}^1$ function $g : \mathbb{R} \to \mathbb{R}$ with at most exponential growth i.e., there exists a positive constant $C$ such that $|g(z)| \le C e^{|z|}$ for any $z \in \mathbb{R}$. Then,
$$\mathbb{E}\big(g(Z)(Z - \mu)  \big) = \sigma^2 \cdot \mathbb{E} \hspace{0.1cm}g'(Z).$$
\end{lemma}

\begin{proof}
Without loss of generality we may assume that $Z \sim \mathcal{N}(0, 1)$ since the case where $Z \sim \mathcal{N}(\mu, \sigma^2)$ will then be straightforward by setting $Z' = \frac{Z-\mu}{\sigma}$. Using integration by parts, we get,
$$\mathbb{E}\big(Z g(Z)\big) = \int_{-\infty}^{+\infty} zg(z) \cdot \frac{e^{-z^2/2}}{\sqrt{2 \pi}} \, \mathrm{d}z = \Big[-g(z) \frac{e^{-z^2/2}}{\sqrt{2 \pi}} \Big]_{z = -\infty}^{z = + \infty} + \int_{-\infty}^{+\infty} g'(z) \cdot \frac{e^{-z^2/2}}{\sqrt{2 \pi}} \, \mathrm{d}z.$$
Owing to the exponential growth assumption, the first term in the right hand side sum is equal to 0. So that, 
$$\mathbb{E}\big(Z g(Z)\big) = \int_{-\infty}^{+\infty} g'(z) \cdot \frac{e^{-z^2/2}}{\sqrt{2 \pi}} \, \mathrm{d}z = \mathbb{E}g'(Z).$$
\end{proof}

\begin{lemma}[See \cite{jourdain:hal-03862241}]
    \label{inf_conv_lemma}
    For any convex function $f : \mathbb{R}^d \to \mathbb{R}$ there exists a sequence $(f_n)_n$ of Lipschitz convex functions such that $f_n \uparrow f$.    
\end{lemma}

\begin{proof}
Let $f : \mathbb{R}^d \to \mathbb{R}$ be a convex function and consider the inf-convolution,
$$f_n(x) := \underset{y \in \mathbb{R}^d}{\inf} \big(f(y) + n|x-y| \big), \quad n \ge 1.$$
Then it is straightforward that for all $x \in \mathbb{R}^d$ and $n \ge 1$,
\begin{equation}
    \label{mono_inf_conv}
    f_n(x) \le f_{n+1}(x) \le f(x).
\end{equation}
Moreover, if we denote by $\nabla_s f(x)$ any sub-gradient of the convex function $f$ at $x$, it follows from the convexity of $f$ that for all $y \in \mathbb{R}^d$:
\begin{align*}
    f(y) + n|y-x| &\ge f(x) + \langle \nabla_s f(x), y-x\rangle + n|x-y|\\
    &\ge f(x) + \big(n - |\nabla_s f(x)|\big)|x-y|
\end{align*}
so that for all $n \ge |\nabla_s f(x)|$, one has $f_n(x) \ge f(x)$. In the second last line, we used Cauchy-Schwartz inequality. Thus combining the latter with Equation \eqref{mono_inf_conv} yields for all $n \ge |\nabla_s f(x)|$, $f_n(x) = f(x)$ so that $f_n \uparrow f$. The convexity of functions $f_n$ is straightforward.
\end{proof}

\begin{lemma}[Characterization of convex ordering \cite{jourdain:hal-03862241}]
\label{rq_carac_ord_cvx}
    In Definition \ref{cvx_ord_def}, other characterizations of convex ordering allow to restrict proofs to Lipschitz convex functions. Indeed, it suffices to consider the inf-convolution of the convex function $f$ defined on $\mathbb{R}^d$ as follows,
        $$f_n(x) := \underset{y \in \mathbb{R}^d}{\inf} \big(f(y) + n|x-y| \big), \quad n \ge 1.$$
        Then, one may show (see Lemma \ref{inf_conv_lemma}) that $f_n$ is a convex function and $f_n \uparrow f$ pointwise. Thus it suffices to check inequality \eqref{cvx_ord_def} for Lipschitz convex functions and obtain the same inequality for convex function as a straightforward application of monotone convergence theorem.
\end{lemma}

\begin{lemma}[Proof of Remark \ref{gen_ex_cnv_matrix}]
\label{proof_lemme_ex_gen_mat_cvx}
Let $\alpha \in [0, 1]$. For all $x \in \mathbb{R}$, set
$$O_{\alpha, x} = O^\top \cdot \diag \big(\overline{\sign}(\lambda_1(x)),\ldots,\overline{\sign}(\lambda_q(x)) \big).$$
Note that, for any $x,y \in \mathbb{R}$, matrix $O_{\alpha, x}$, $O_{\alpha, y}$ thus defined are orthogonal as a product of two orthogonal matrix. Then, by simple algebra, one has
\begin{align*}
    &\big(\alpha \sigma(x)O_{\alpha, x} + (1-\alpha)\sigma(y)O_{\alpha, y}\big)\big(\alpha \sigma(x)O_{\alpha, x} + (1-\alpha)\sigma(y)O_{\alpha, y}\big)^\top - \sigma \sigma^\top\big(\alpha x + (1-\alpha)y \big)\\
    &= A \Big[\Tilde{D}_{\alpha}(x, y) - D(\alpha x + (1-\alpha)y)^2 \Big] A^\top,
\end{align*}
where the diagonal matrix $\Tilde{D}_{\alpha}(x, y) \in \mathbb{M}_{q,q}(\mathbb{R})$ is given by,
$$\Tilde{D}_{\alpha}(x, y) := \alpha^2 D^2(x) + \alpha(1-\alpha)D(x)O O_{\alpha, x}O_{\alpha, y}^\top O^\top D(y) + \alpha(1-\alpha)D(y)O O_{\alpha, y}O_{\alpha, x}^\top O^\top D(x) + (1-\alpha)^2D^2(y).$$
But, using simple algebra and the definition of matrix $O_{\alpha, x}$, $O_{\alpha, y}$, one has
$$D(x)O O_{\alpha, x}O_{\alpha, y}^\top O^\top D(y) = \overline{D}(x) \cdot \overline{D}(y) \quad \text{and} \quad D(y)O O_{\alpha, y}O_{\alpha, x}^\top O^\top D(x) = \overline{D}(x) \cdot \overline{D}(y)$$
with
$$\overline{D}(x) := \diag\big(|\lambda_1(x)|, \ldots, |\lambda_q(x)| \big).$$
Thus the diagonal matrix $\Tilde{D}_{\alpha}(x, y)$ reads,
$$\Tilde{D}_{\alpha}(x, y) = \alpha^2 D^2(x) + 2\alpha(1-\alpha)\overline{D}(x) \cdot \overline{D}(y) + (1-\alpha)^2D^2(y)$$
and the diagonal matrix $\Tilde{D}_{\alpha}(x, y) - D(\alpha x + (1-\alpha)y)^2$ has non-negative diagonal entries owing to the convexity of all $|\lambda_i|$. This prove that $\Tilde{D}_{\alpha}(x, y) - D(\alpha x + (1-\alpha)y)^2 \in \mathcal{S}^{+}\big(d, \mathbb{R}\big)$ and as a straightforward consequence that $A \Big[\Tilde{D}_{\alpha}(x, y) - D(\alpha x + (1-\alpha)y)^2 \Big] A^\top \in \mathcal{S}^{+}\big(d, \mathbb{R}\big)$. This completes the proof.
\end{lemma}

\section{Background on convex ordering (Proofs)}
\label{appendix_proof_some_prop}
\begin{proof}[Proof of Proposition \ref{prop_first_prop_cvx_ord}]
\label{proof_first_result_cvx_ord}
We only prove (B). First notice that if $U = \mathbb{E}(V \rvert U)$ then one has for every convex function $f : \mathbb{R}^d \to \mathbb{R}$ by applying (conditional) Jensen's inequality and then the law of iterated expectations, one has
$$\mathbb{E}f(U) = \mathbb{E}\big(f\big(\mathbb{E}\big(V \rvert U \big) \big) \big) \le \mathbb{E}\big( \mathbb{E}\big(f(V) \rvert U \big)\big) = \mathbb{E}f(V)$$
so that $U \preceq_{cvx} V$. Besides, let $Z_1, Z_2 \sim \mathcal{N}(0,\mathbf{I}_q)$ be independent random vectors. We define
$$U = AZ_1 \quad \quad \text{and} \quad \quad V = U + \big(BB^\top - AA^\top \big)^{1/2} Z_2,$$
where for some $d \times d$ real matrix $M \in \mathcal{S}^{+}\big(d, \mathbb{R}\big)$, the matrix $M^{1/2}$ is defined and satisfy $M^{1/2} \cdot \big(M^{1/2}\big)^\top = M$. We have $V \sim \mathcal{N}\big(0,AA^\top + \big(\big(BB^\top - AA^\top \big)^{1/2}\big)^2 \big) = \mathcal{N}\big(0,BB^\top \big)$ and the result follows by noticing that $U = \mathbb{E}(V \rvert U)$.
\end{proof}

\begin{proof}[Proof of Proposition \ref{useful_properties_cvx_ord}]
Let $f : \mathbb{R}^d \to \mathbb{R}$ be a convex function.
\begin{enumerate}[label=\roman*.]
\item For $x, y \in \mathbb{R}^d, \lambda \in [0, 1]$, define $\Bar{xy}^{\lambda} := \lambda x + (1-\lambda) y$. Using successively the $\preceq$-convexity of $\sigma_{t_{\ell}^{(mn)}}$, Proposition \ref{gen_radial_distri}, and the convexity of $f$, we get:
\begin{align*}
    \mathrm{P}^{(mn)}_\ell(f)\big(\Bar{xy}^{\lambda}\big) &= \mathbb{E}f\big(\Bar{xy}^{\lambda} + h \cdot \kappa(t_{\ell}^{(mn)}) (\Bar{xy}^{\lambda}-\zeta) + \sqrt{h} \cdot \sigma_{t_{\ell}^{(mn)}}(\Bar{xy}^{\lambda})Z  \big)\\
    &\le \mathbb{E}f\big(\Bar{xy}^{\lambda} + h \cdot \kappa(t_{\ell}^{(mn)}) (\Bar{xy}^{\lambda}-\zeta) + \sqrt{h} \cdot (\lambda \sigma_{t_{\ell}^{(mn)}}(x)Z + (1-\lambda)\sigma_{t_{\ell}^{(mn)}}(y) Z)  \big)\\
    &\le \lambda \mathrm{P}^{(mn)}_\ell(f)(x) + (1-\lambda) \mathrm{P}^{(mn)}_\ell(f)(y).
\end{align*}

\item This is a straightforward application of Proposition \ref{gen_radial_distri}.
\end{enumerate}
\end{proof}

\section{Some results on matrix \texorpdfstring{$\Gamma$}{}}
\label{chol_decomp}
We focus on the correlation matrix $\Gamma$ \eqref{mat_gamma}.
\begin{Proposition}[Explicit Cholesky decomposition]
\label{chol_dec_prop}
Consider the definite positive matrix $\Gamma(\rho)$ defined in \eqref{mat_gamma} with $\rho \in \big(-\frac{1}{q - 1}, 1\big)$. Its Cholesky decomposition is given by
\begin{equation}
\label{chol_decomp_mat_gamma}
L(\rho) = \begin{pmatrix} 
         d_1 &0& 0 & \cdots & 0 \\ 
         \ell_1 & d_2 & 0 & \cdots & 0\\
         \ell_1 & \ell_2 & d_3 & \cdots & 0\\ 
         \vdots & &  & \ddots &  \\ 
         \ell_1 & \ell_2 & \ell_3 & \cdots & d_q
\end{pmatrix}\in \mathbb{M}_{q,q}\big(\mathbb{R}\big),
\end{equation}
where $d_1 = 1, \ell_1 = \rho$ and for any $j \ge 2$,
$$d_j = \sqrt{d_{j-1}^2 - \ell_{j-1}^2} \hspace{0.6cm} \text{and} \hspace{0.6cm} \ell_j = \frac{\rho - 1}{d_j} + d_j.$$
\end{Proposition}

\begin{proof}
We proceed by an induction on $q$. It is straightforward that the result holds for $q = 2$. Assume it holds for some $q$. For convenience we use $\Gamma_q$ instead of $\Gamma(\rho)$ to specify that $\Gamma(\rho)$ lies in $\mathbb{M}_{q, q}\big(\mathbb{R}\big)$. By the induction assumption and using block partition of $\Gamma_q$, we have,
\begin{equation}
\label{hyp_rec_mat_chol}
\Gamma_q = \begin{pmatrix} 
         L_{q-1} &0\\ 
         v_q & d_q
\end{pmatrix} \times \begin{pmatrix} 
         L_{q-1}^\top & v_q^\top\\ 
         0 & d_q
\end{pmatrix} = \begin{pmatrix} 
         L_{q-1} \cdot L_{q-1}^\top & L_{q-1} \cdot v_q^\top\\ 
         v_q \cdot L_{q-1}^\top & v_q \cdot v_q^\top + d_q^2
\end{pmatrix},
\end{equation}
where $v_q = \big(\ell_1, \ldots, \ell_{q-1} \big)$. Then equating the bottom right element of these two matrices gives,
\begin{equation}
\label{first_res}
v_q \cdot v_q^\top + d_q^2 = 1.
\end{equation}

Using \eqref{first_res} which implies $v_q \cdot v_q^\top = 1 - d_q^2 $ and then identifying the off-diagonal elements in the final column yields, for $1 \le i < q$, $\rho = \big[L_{q-1} \cdot v_q^\top \big]_i$ so that,
\begin{equation}
\label{second_res}
L_{q-1} \cdot v_q^\top = (\rho, \ldots, \rho).
\end{equation}

Let us come back to our purpose which is to show that \eqref{hyp_rec_mat_chol} holds for $q+1$. It suffices to prove that $L_{q} \cdot v_{q+1}^\top = (\rho, \ldots, \rho)$ and $v_{q+1} \cdot v_{q+1}^\top + d_{q+1}^2 = 1$. But,
\begin{equation*}
L_{q} \cdot v_{q+1}^\top = \begin{pmatrix} 
         L_{q-1} & 0\\ 
         v_{q} & d_q
\end{pmatrix} \times \begin{pmatrix} 
         v_{q} ^\top\\ 
         \ell_{q}
\end{pmatrix} = \begin{pmatrix} 
         L_{q-1} \cdot v_{q} ^\top\\ 
         v_{q} \cdot v_{q}^\top + d_{q} \cdot \ell_q
\end{pmatrix}.
\end{equation*}
Moreover, it follows from \eqref{first_res} that $ v_{q} \cdot v_{q}^\top + d_{q} \cdot \ell_q = 1 - d_q^2 + d_q \cdot \big( \frac{\rho - 1}{d_q} + d_q \big) = \rho$. Combined with \eqref{second_res} implies that $L_{q} \cdot v_{q+1}^\top = (\rho, \ldots, \rho)$. Finally, it follows from the block partition of $v_{q+1}$ and equation \eqref{first_res},
$$v_{q+1} \cdot v_{q+1}^\top + d_{q+1}^2 = v_{q} \cdot v_{q}^\top + \ell_q ^2+ d_{q+1}^2 = 1 - d_q^2 + \ell_q^2 + d_q^2 - \ell_q^2 = 1.$$
This completes the proof.
\end{proof}

\begin{Proposition}[$\preceq$-monotony in $\rho$]
\label{mono_corr}
Let $\rho_1, \rho_2 \in \big(-\frac{1}{q-1}, 1 \big)$ such that $\rho_1 \le \rho_2$. For any $p \in \{1,2\}$, consider the matrix-valued field,
$$\sigma_{\rho_p}(x) := \big(\lambda_1(x), \ldots, \lambda_q(x) \big) \cdot L(\rho_p) \in \mathbb{M}_{1, q}\big(\mathbb{R}\big),$$
for some non-negative (real) functions $\lambda_i$ ($1 \le i \le q$) and where $L(\rho_p)$ denotes the Cholesky decomposition of the correlation matrix $\Gamma(\rho_p)$ i.e., the matrix $\Gamma$ in \eqref{mat_gamma} associated with the correlation parameter $\rho_p$. Then we have,
$$\sigma_{\rho_1}(x) \preceq \sigma_{\rho_2}(x).$$
\end{Proposition}

\begin{proof}
Note that, 
\begin{align*}
\sigma_{\rho_p}\sigma^\top_{\rho_p}(x) &= \big(\lambda_1(x), \ldots, \lambda_q(x) \big) \cdot \Gamma(\rho_p) \cdot \big(\lambda_1(x), \ldots, \lambda_q(x) \big)^\top\\
&= \sum_{i = 1}^{q} \sum_{j = 1}^{q} \Gamma_{i, j}(\rho_p) \lambda_i(x) \lambda_j(x) = \sum_{i = 1}^{q} \lambda_i(x)^2 + \rho_p \cdot \sum_{1 \le i \neq j \le q} \lambda_i(x) \lambda_j(x).
\end{align*}
Thus, since $\rho_1 \le \rho_2$ and $\lambda_i(x)$ are non-negative, then $\sigma_{\rho_1}\sigma^\top_{\rho_1}(x) \le \sigma_{\rho_2}\sigma^\top_{\rho_2}(x)$ so that $\sigma_{\rho_1}(\cdot) \preceq \sigma_{\rho_2}(\cdot)$. This completes the proof.
\end{proof}

\begin{Proposition}
\label{vol_mono}
Let $A, B \in \mathbb{M}_{1, q}\big(\mathbb{R}\big)$. Consider the correlation matrix $\Gamma(\rho)$ in \eqref{mat_gamma} and denote by $L(\rho)$ its Cholesky decomposition given by \eqref{chol_decomp_mat_gamma}. We have the following results.
\begin{countlist}[label={(\alph*)}]{otherlist2}
  \item \label{gen_dom_vol} If,
  \begin{equation}
  \label{eq_dom_gen_vol}
  \sum_{i = 1}^{q}\sum_{j = 1}^{q} \big(B_{1,i} B_{1,j} - A_{1,i} A_{1,j} \big) \cdot \Gamma_{i, j}(\rho) \ge 0,
  \end{equation}
  then, we have $AL(\rho) \preceq BL(\rho)$.
  
  \item In particular, if $\rho \in [0,1)$, then $ |A|^2 \le |B|^2 \implies AL \preceq BL$.
\end{countlist}
\end{Proposition}

\begin{proof}
We drop the variable $\rho$ for simplicity.

\begin{countlist}[label={(\alph*)}]{otherlist3}
  \item Since $LL^\top = \Gamma$, we have,
  $$\big(BL\big)\big(BL\big)^\top - \big(AL\big)\big(AL\big)^\top = B \Gamma  B^\top - A \Gamma A^\top = \sum_{i = 1}^{q}\sum_{j = 1}^{q} \big(B_{1,i} B_{1,j} - A_{1,i} A_{1,j} \big) \cdot \Gamma_{i, j}.$$
  Thus condition \eqref{eq_dom_gen_vol} implies that $0 \le A \Gamma A^\top \le B \Gamma B^\top$ which yields $AL \preceq BL$.
  
  \item Assume that $\rho \in [0, 1)$. Using the definition of the correlation matrix $\Gamma$, one has:
  $$\sum_{i = 1}^{q}\sum_{j = 1}^{q} \big(B_{1,i} B_{1,j} - A_{1,i} A_{1,j} \big) \cdot \Gamma_{i, j} = \rho \Big(\sum_{i = 1}^q B_{1,i} - A_{1,i}  \Big)^2 + (1-\rho) \sum_{i = 1}^q B_{1,i}^2 - A_{1,i}^2$$
  which is non negative if $\rho \in [0, 1)$ and $|A|^2 \le |B|^2$. This completes the proof.
\end{countlist}
\end{proof}
\end{document}